\documentclass[final,onefignum,onetabnum]{siamart190516}

\usepackage{braket,amsfonts}
\usepackage{array}
\usepackage[caption=false]{subfig}
\usepackage{pgfplots}
\usepackage{amssymb,color}
\usepackage{graphicx,epsfig}
\usepackage{epstopdf}
\usepackage{amsmath}
\usepackage{enumerate}
\usepackage{mathrsfs}
\usepackage{array}
\usepackage{ntheorem}

\newsiamthm{claim}{Claim}
\newsiamthm{obs}{Observation}
\newsiamthm{lem}{Lemma}
\newsiamthm{remark}{Remark}

\title{Deep Learning-Based Solvability of Underdetermined Inverse Problems in Medical Imaging \thanks{This work was supported by Samsung Science $\&$ Technology Foundation (No. SRFC-IT1902-09).}}
\author{Chang Min Hyun\thanks{Department of Computational Science and Engineering, Yonsei University, Seoul, Korea} \and
	Seong Hyeon Baek\footnotemark[2] \and Mingyu Lee\footnotemark[2] \and
	Sung Min Lee\footnotemark[2] \and
	Jin Keun Seo\footnotemark[2] \thanks{To whom the corresponding author  (\email{seoj@yonsei.ac.kr})}}

\headers{Underdetermined Inverse Problems}{C. M. Hyun, S. H. Baek, M. Lee, S. M. Lee, and J. K. Seo}

\def\R{{\Bbb R}}

\def\Om{\Omega}

\def\f{\frac}
\def\p{\partial}
\def\q{\quad}

\def\b{\mathbf{b}}
\def\c{\mathbf{c}}
\def\d{\mathbf{d}}

\def\h{{\mathbf h}}

\def\u{\mathbf{u}}

\def\x{\mathbf{x}}
\def\y{\mathbf{y}}
\def\z{\mathbf{z}}
\def\A{{\mathbf A}}
\def\B{\mathbf{B}}
\def\D{{\mathbf D}}

\def\mR{{\mathcal R}}

\def\bt{\boldsymbol{t}}
\def\roi{{\Omega_{\mbox{\tiny ROI}}}}
\def\fb{{f_{\flat}}}
\def\bth{\boldsymbol{\theta}}

\begin{document}
\maketitle

\begin{abstract}
	Recently, with the significant developments in deep learning techniques, solving underdetermined inverse problems has become one of the major concerns in the medical imaging domain, where underdetermined problems are motivated by the willingness to provide high resolution medical images with as little data as possible, by optimizing data collection in terms of minimal acquisition time, cost-effectiveness, and low invasiveness.  Typical examples include undersampled magnetic resonance imaging(MRI), interior tomography, and sparse-view computed tomography(CT), where deep learning techniques have achieved excellent performances. However, there is a lack of mathematical analysis of why the deep learning method is performing well. This study aims to explain about learning the causal relationship regarding the structure of the training data suitable for deep learning, to solve highly underdetermined problems. We present a particular low-dimensional solution model to highlight the advantage of deep learning methods over conventional methods, where two approaches use the prior information of the solution in a completely different way. We also analyze whether deep learning methods can learn the desired reconstruction map from training data in the three models (undersampled MRI, sparse-view CT, interior tomography). This paper also discusses the nonlinearity structure of underdetermined linear systems and conditions of learning (called $\mathcal{M}$-RIP condition).
\end{abstract}



\section{Introduction}
In medical imaging, we want to improve our visual ability to provide meaningful expression and useful description of diagnosis and treatment, while optimizing data collection in terms of minimal acquisition time, cost-effectiveness, and low invasiveness. The goal is to find a function $f$ that maps from inputs (what we measure) to outputs (reconstructing useful medical images):
\begin{equation}\label{MI-f}
f(\mbox{input data}) = \mbox{useful output}.
\end{equation}
The output could be the two- or three-dimensional visual representation of the interior of a body, such as computerized tomography (CT)  \cite{Hounsfield1973,Natterer1986} and magnetic resonance imaging (MRI)  \cite{Lauterbur1973,Haacke1999}. To achieve the output of reasonable resolution and accuracy, we have used a suitable means of measurement (input), which allows to reconstruct the output.

Conventional CT and MRI data collections are designed to obtain a well-posed reconstruction method in the sense that the corresponding forward models are well-posed. The forward model can be expressed as the well-posed linear system as follows:
\begin{equation} \label{fullAyb}
\A_{\mbox{\tiny{full}}}\y =\b_{\mbox{\tiny{full}}}
\end{equation}
where $\b_{\mbox{\tiny{full}}}$ denotes the ``fully sampled" data (e.g, sinogram in CT and k-space data in MRI), $\y$ denotes the CT or MRI image, and $\textbf{A}_{\mbox{\tiny{full}}}$ denotes the invertible matrix representing discrete Radon transform for CT and discrete Fourier transform for MRI.
More precisely, the forward model of CT is based on the assumption that the measured X-ray projection data is the Radon transform of the output image. To achieve the given resolution of the output and invert the corresponding discrete Radon transform, a sufficient number of projection angles are required so that the number of equations (measured data) becomes greater than the number of unknowns (the number of pixels in the image). Hence, given the resolution of the output, the traditional approaches require a certain amount of data acquisition in order to make the reconstruction problem  well-posed in the sense of Hadamard\cite{Hadamard1902}.

Because of the great needs to reduce the radiation dose in CT and data acquisition time in MRI, considerable attention has been given to solve underdetermined problems (or ill-posed inverse problems) that violate the Nyquist criteria \cite{Nyquist1928} in the sense that the number of equations is much smaller than the number of unknowns. The demand for undersampled MRI is because of the long scan time with the human body trapped in an inconvenient narrow bore; shortening the MRI scan time can increase the satisfaction of patient, reduce the artifacts caused by patient movement, and reduce the medical costs \cite{Koff2006,Lustig2007,Candes2006,Donoho2006,Sodickson1997}. The need for low-dose CT arises from the cancer risks associated with the exposure of patients to ionizing radiation \cite{Islam2006,Kudo2013,Brenner2017}. A highly underdetermined problem (far less equations than unknowns) corresponding to \eqref{fullAyb} can be expressed as follows:
\begin{equation} \label{undersampled}
\underbrace{\A}_{\mathcal S_{\mbox{\tiny{sub}}}(\A_{\mbox{\tiny{full}}})}\y =\underbrace{\b}_{ \mathcal S_{\mbox{\tiny{sub}}}(\b_{\mbox{\tiny{full}}})}
\end{equation}
where  $\mathcal S_{\mbox{\tiny{sub}}}$ denotes a subsampling operator. For example, in undersampled MRI, $\b=\mathcal S_{\mbox{\tiny{sub}}}(\b_{\mbox{\tiny{full}}})$ denotes an undersampled k-space data violating the Nyquist sampling criterion, and $\y$ denotes MRI image reconstructed using a fully sampled k-space data. Because $\A$ is not an invertible matrix, there exist infinitely many solutions.

Solving the underdetermined problem \eqref{undersampled} depends on the appropriate use of {\it a priori} information about medical CT or MRI images as solutions. However, the conventional approaches using prior knowledge, such as regularization and compressed sensing(CS) approaches, may not be appropriate for medical images in which small anomalous details are more important than the overall feature \cite{Jaspan2015}. Currently, deep learning techniques have exhibited excellent achievement in various underdetermined problems such as undersampled MRI, interior tomography, and sparse view CT. They seem to overcome the limitations of the existing mathematical methods in handling various ill-posed problems \cite{Jin2017,Han2017,Hyun2018}.
It is highly expected that deep learning methodologies will improve their performance, as training data and experience are accumulated over time. However, there is a tremendous lack of a rigorous mathematical foundation that would allow us to understand the reasons for the remarkable performance of deep learning methods \cite{Kawaguchi2017}.

This study aims to provide a systematic basis for learning the causal relationship regarding the structure of the training data suitable for deep learning to solve highly underdetermined problems. The goal of the undersampled problem \eqref{undersampled} is to find a reconstruction map $\fb:\b \to \y$ that maps from the highly undersampled data $\b$ to the image $\y=\A_{\mbox{\tiny{full}}}^{-1}\b_{\mbox{\tiny{full}}}$ in \eqref{fullAyb}. Here, for ease of explanation, we ignore the noise in $\b$ and abuse the notation of $\A_{\mbox{\tiny{full}}}^{-1}$, which should be understood as representing the filtered backprojection (FBP) in CT and the inverse Fourier transform in MRI \cite{Seo2013}.  Without using a constraint on $\y$, one cannot find the reconstruction map $\fb$.  Hence, we must use the prior knowledge on the data distribution of all possible images to be reconstructed.
To extract prior knowledge on solutions, deep learning-based techniques use training data $\{(\b^{(k)},\y^{(k)})\}_{k=1}^{\mathfrak{n}_{\mbox{\tiny data}}}$, where $\y^{(k)}= \A_{\mbox{\tiny{full}}}^{-1} \b^{(k)}_{\mbox{\tiny{full}}}$ and $\b^{(k)}=\mathcal{S}_{\mbox{\tiny sub}}\b^{(k)}_{\mbox{\tiny{full}}}$.

Learning $\fb:\b\to \y$ can be achieved by learning the following map:
\begin{equation} \label{f-DL1}
f: \x= \A^{\sharp}\b \mapsto \y=\A_{\mbox{\tiny full}}^{-1}\b_{\mbox{\tiny full}}.
\end{equation}
where $\A^{\sharp} = \A_{\mbox{\tiny full}}^{-1} \mathcal S_{\mbox{\tiny sub}}^*$ and $\mathcal S_{\mbox{\tiny sub}}^*$ is the dual of $\mathcal S_{\mbox{\tiny sub}}$, which can be understood as the zero padding operator corresponding to the subsampling $\mathcal S_{\mbox{\tiny sub}}$. Using the transformed training data $\{(\x^{(k)},\y^{(k)})\}_{k=1}^{\mathfrak{n}_{\mbox{\tiny data}}}$, where $\x^{(k)}=\A^\sharp\b^{(k)}$, we consider the learning objective as follows:
\begin{equation}\label{MDL}
f= \underset{f \in \mathbb{NN}}{\mbox{argmin}} ~ \sum_{k=1}^{\mathfrak{n}_{\mbox{\tiny data}}} \| f(\x^{(k)}) - \y^{(k)} \|_{\ell_2}
\end{equation}
where $\Bbb{NN}$ denotes a set of functions described in a special form of neural network and $\|\cdot\|_{\ell_p}$ is the $\ell_p$ norm of the vector. Notably, this $f$ in \eqref{MDL} is designed to work well only on a low-dimensional solution manifold obtained by regressing the training data $\{\y^{(k)}\}_{k=1}^{\mathfrak{n}_{\mbox{\tiny data}}}$, not on the entire image domain.

This paper aims to provide some mathematical grounds for the learnability of $f$ by using various performance experiments. In Section \ref{M-example},  we present a particular low-dimensional solution model to highlight the advantage of a deep learning method over conventional methods using PCA, wavelets, and total variation regularization. In Section \ref{sec-Solvability}, we discusses the nonlinearity structure of underdetermined linear systems and conditions of learning (called $\mathcal{M}$-RIP condition). We observe that highly underdetermined linear systems in medical imaging are highly non-linear. We also examine whether a desired reconstruction map $f:\x \to \y$ can be learnable from the training data. The learning ability depends on the subsampling strategy $\mathcal{S}_{\mbox{\tiny sub}}$, and quality and quantity of training data. Section \ref{SolveUMRI} investigates the learnability of undersampled MRI. It depends on the sampling pattern as follows: (i) If $\mathcal{S}_{\mbox{\tiny sub}}$ is a uniform subsampling, $f$ is not learnable because there exist two different realistic images, namely, $\y$ and $\y'$, such that $\A^\sharp \A(\y-\y')=0$. (ii) If $\mathcal{S}_{\mbox{\tiny sub}}$ denotes a uniform sampling with one additional phase encoding line, then $f$ is learnable. We also deal with the learnability of $f$ in interior tomography (see Section \ref{sec-localCT}) and  sparse-view CT (see Section \ref{sec-sparseCT}). In interior tomography, $f$ is learnable because $f(\x)-\x$ is directionally analytic; therefore, it is determined by the very local information of it. In sparse-view CT, $f$ is somehow learnable because $f(\x)-\x$ has common repetitive local patterns that are very different from realistic images. Finally, in Section \ref{sec-discussion}, we discuss some issues related to deep learning-based solvability for underdetermined problems.

\section{Analysis on Underdetermined Inverse Problem}
\begin{figure*}[t!]
\begin{center}
	\includegraphics[width=1\textwidth]{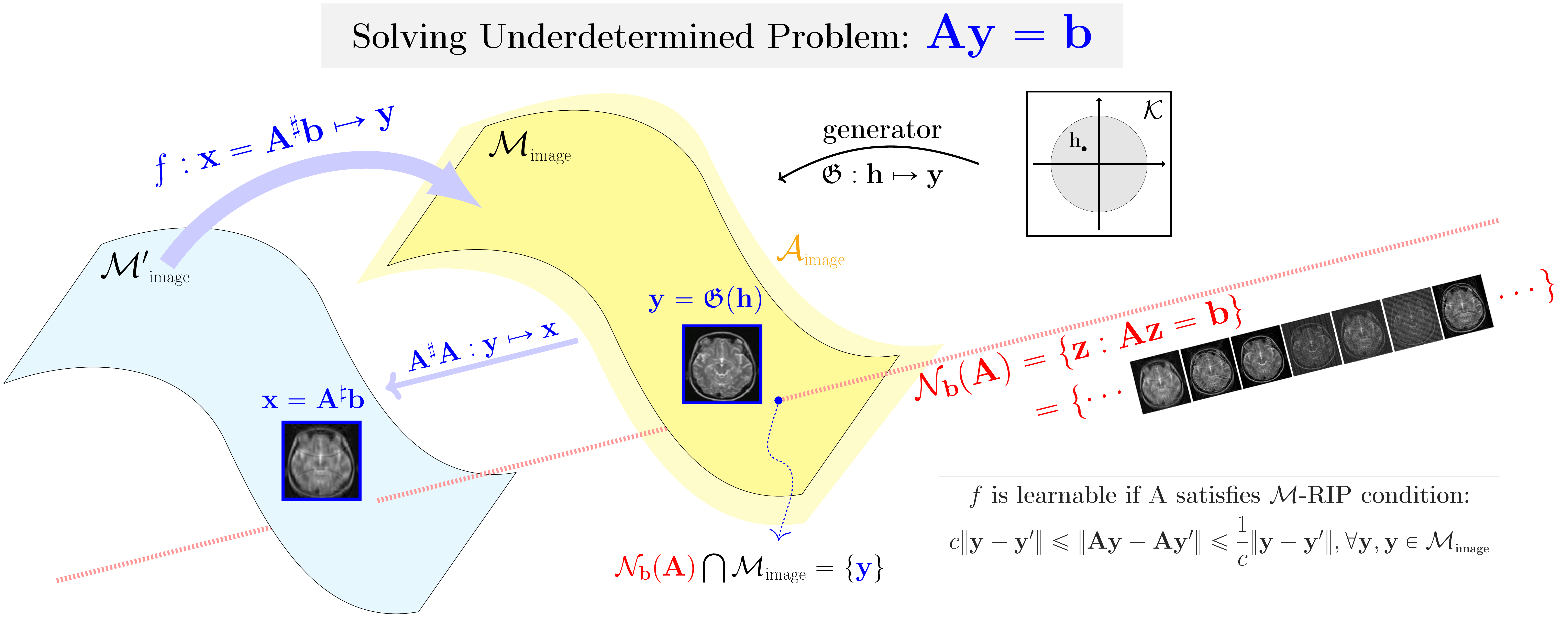}
	\caption{Description of solvability of underdetermined inverse problem $\A \y = \b$. Solving $\A \y = \b$ can be achieved by learning $f : \x = \A^\sharp \b \mapsto \y = \A_{\mbox{\tiny full}}^{-1}\b_{\mbox{\tiny full}}$ in \eqref{f-DL1} with probing the solution manifold $\mathcal M_{\mbox{\tiny image}}$. If $\A$ satisfies the $\mathcal M$-RIP condition, then $\A^\sharp \A : \mathcal M_{\mbox{\tiny image}} \mapsto \mathcal M_{\mbox{\tiny image}}^\prime$ is one-to-one, i.e., $\mathcal{N}_\b(\A) \cap \mathcal M_{\mbox{\tiny image}}=\{ \y \}$ is unique. In general, $f$ is nonlinear and the degree of non-linearity depends on the sampling strategy of $\b$ and the degree of bending the solution manifold.}
	\label{Fig-main diagram}
\end{center}
\end{figure*}
This section considers the underdetermined problem \eqref{undersampled}, where $\b$ represents the undersampled data (e.g., $k$-space data in undersampled MRI and sinogram data in sparse-view CT and interior tomography).  We denote the dimensions of row and column vectors by $n$ and $m$, respectively. Note that $n$ is the same as the dimension of image. The relation between the undersampled data $\b$ and the corresponding fully sampled data $\b_{\mbox{\tiny{full}}}$ can be expressed by
\begin{equation}\label{subsample}
\b=\mathcal S_{\mbox{\tiny{sub}}}(\b_{\mbox{\tiny{full}}})
\end{equation}
where $\mathcal S_{\mbox{\tiny{sub}}}$ denotes the subsampling operator.
Since $\A$ is $m \times n$ matrix with $m \ll n$, the underdetermined problem \eqref{undersampled} has infinitely many solutions, which constitute the $n-m$ dimensional subplane given by the followings:
\begin{equation}\label{NbA}
\mathcal N_b(\A):=\{\y\in \mathbb{R}^{n}(\mbox{or } \mathbb{C}^n):  \A\y=\b \}
\end{equation}
To find $f$ in \eqref{f-DL1}, it is necessary to find a way to convert the distorted image $\x=\A^\sharp\b$ to the desired image $\y$, which is selected from the set $\mathcal N_b(\A)$.
In order to determine the unique solution among $\mathcal N_b(\A)$, we have to restrict the solution by invoking the prior knowledge of expected solutions.

\subsection{Constrained reconstruction problem}
Assume that $\mathcal{A}_{\mbox{\tiny image}}$ is a set of all realistic images that include the set of all $\y=\A_{\mbox{\tiny full}}^{-1}\b_{\mbox{\tiny full}}\in \mathbb{R}^{n}$(or $\mathbb{C}^{n}$), where $\b_{\mbox{\tiny full}}$ denotes the fully sampled medical data. We consider the following constraint problem:
\begin{equation} \label{ConMin}
\left\{ \begin{array}{l}
\mbox{Solve}  ~~  \A \y = \b \\
\mbox{subject to the constrait } \y\in \mathcal{A}_{\mbox{\tiny image}}.
\end{array}
\right.
\end{equation}
Ideally, we hope that $\mathcal N_b(\A) \cap \mathcal{A}_{\mbox{\tiny image}}
=\{ \y\in \mathcal{A}_{\mbox{\tiny image}}: \A\y=\b\}\neq \emptyset$, and that all the images in the set $\mathcal N_b(\A) \cap \mathcal{A}_{\mbox{\tiny image}}$ are visually same for radiologists. Hence, it seems to be necessary to describe a similarity measure between two images, $\y, \y'\in \mathcal{A}_{\mbox{\tiny image}}$, by defining the distance $\mbox{dist}_{\mbox{\tiny radiologist}}(\y, \y')$; e.g., $\mbox{dist}_{\mbox{\tiny radiologist}}(\y, \y')= 0$  means that both images are visually the same for radiologists. Currently, it seems to be considerably difficult to develop a concept of $\mbox{dist}_{\mbox{\tiny radiologist}}(\y, \y')$ that agrees with the perspective of medical radiologists. To simplify the problem \eqref{ConMin} along with avoiding complex similarity issues in terms of radiologists, let us assume the following:
\begin{itemize}
\item[\mbox{[H1]}] Any image in $ \mathcal{A}_{\mbox{\tiny image}}$ lies on or near a low-dimensional manifold, which is denoted by $\mathcal{M}_{\mbox{\tiny image}}$, whose Hausdorff dimension, which is denoted by $\mathfrak{d}_{\mbox{\tiny mfd}}$, is smaller than $m$ (i.e., the dimension of sampled vector $\b$).
\item[\mbox{[H2]}] There exists a generator $\mathfrak{G}: \h\in \mathcal K \to \y\in \R^{n}$ such that the following hold:
\begin{equation}
\mathcal{M}_{\mbox{\tiny image}} = \{ \textbf{y} \in \mathbb{R}^{n} : \textbf{y}=\mathfrak{G}(\h) ~\mbox{and}~ \h \in \mathcal K \}
\end{equation}
where $\mathcal K$ denotes a subset of $\R^{\mathfrak{d}_{\mbox{\tiny mfd}}}$. Moreover, there exists a constant $c \in (0,1]$ such that the following hold: For all $\h,\h' \in \mathcal K$,
\begin{equation}
c \| \h-\h'\| \le  \| \mathfrak{G}(\h) - \mathfrak{G}(\h')\| \le \f{1}{c} \|\h-\h'\|
\end{equation}
\item[\mbox{[H3]}] There exists a normalization map $\mathfrak N : \mathcal{A}_{\mbox{\tiny image}} \mapsto \mathcal{M}_{\mbox{\tiny image}}$ such that if two images $\y, \y'\in\mathcal{A}_{\mbox{\tiny image}}$ are visually the same for radiologists, then $\mathfrak N(\y)=\mathfrak N(\y')$.
\end{itemize}
The manifold $\mathcal{M}_{\mbox{\tiny image}}$ can be viewed as a set of all $256\times 256$ human head-MR images in undersampled MRI problems, or as a set of all $512\times 512$ CT images in underdetermined CT problems. In the normalization map $\mathfrak N$, the difference $\y-\mathfrak N(\y)$ can be a noise that does not contain any diagnostic feature.

\begin{figure*}[t!]
\centering
\includegraphics[width=0.9\textwidth]{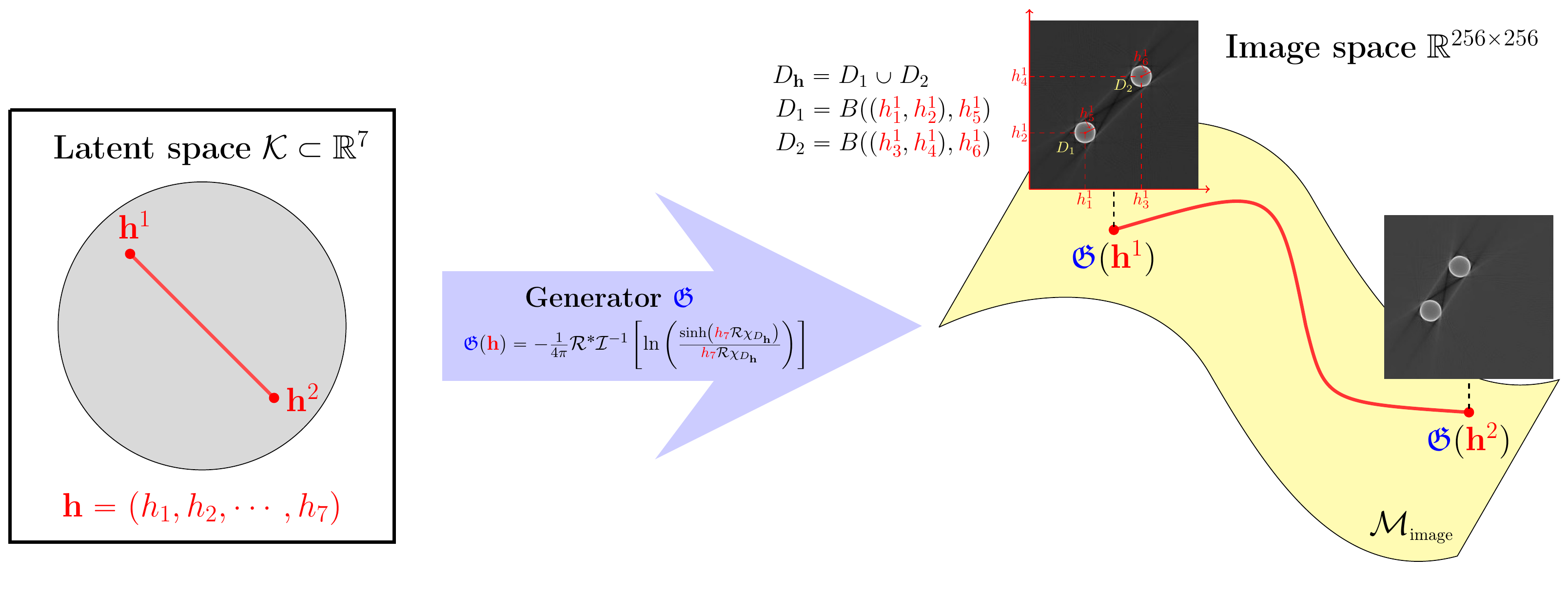}
\caption{Special solution manifold $\mathcal{M}_{\mbox{\tiny image}} \subset \mathbb{R}^{256 \times 256}$ to highlight the advantage of deep learning method over conventional methods. It is generated by the generator function $\mathfrak{G}$ in \eqref{Mani2}. The manifold $\mathcal{M}_{\mbox{\tiny image}}$ is seven dimensional. Images in $\mathcal{M}_{\mbox{\tiny image}}$ consist of two disks and a special streaking feature joining the two disks.   }
\label{FigMA}
\end{figure*}

If we have both the generator $\mathfrak{G}$ and the normalizer $\mathfrak N$ in the above assumptions [H1]--[H3], the underdetermined problem \eqref{ConMin} becomes a somewhat well-posed problem as follows:
\begin{equation}\label{ConMin2}
\mbox{Given $\x=\A^\sharp\b$, solve}  ~ \A\mathfrak G (\h)  = \A \x ~ \mbox{for}~\h
\end{equation}
where the number of unknowns are smaller than the number of equations. A necessary condition for the solvability of \eqref{ConMin2} is $\mathfrak{d}_{\mbox{\tiny mfd}}\le  n-m$.
Moreover, with the aid of the generator $\mathfrak{G}$, the very ambiguous distance $
\mbox{dist}_{\mbox{\tiny radiologist}}(\y, \y')$ from the viewpoint of radiologist can be clearly defined as $\| \h - \h^\prime \|$, where $\mathfrak G (\h) = \mathfrak{N}(\y)$ and $\mathfrak G(\h^\prime) = \mathfrak{N}(\y^\prime)$. However, finding both the generator $\mathfrak{G}$ and the normalizer $\mathfrak N$ may be very difficult task, which is expected to be achieved via deep learning techniques using a training dataset $\{\y^{(k)}\}_{k=1}^{\mathfrak{n}_{\mbox{\tiny data}}}$ in the near future.

The reconstruction map $f: \x\to \y$ in \eqref{f-DL1} can be expressed as
\begin{equation}\label{f-Mani-min}
f(\x):=\underset{\textbf{y}\in \mathcal{M}_{\mbox{\tiny image}}}{\mbox{argmin}} ~ \| \A^\sharp \A \y - \x \|^2_{\ell^2},
\end{equation}
by assuming that there exists a unique minimizer and that $\mathcal{M}_{\mbox{\tiny image}}$ is known. Since it is very difficult to know the manifold $\mathcal{M}_{\mbox{\tiny image}}$, one can achieve the reconstruction map $f$ as follows:
\begin{equation}\label{f-NN-min}
f :=\underset{f\in \Bbb{NN}}{\mbox{argmin}} ~ \sum_{k=1}^{\mathfrak{n}_{\mbox{\tiny data}}} \| f(\x^{(k)})- \y^{(k)} \|^2_{\ell^2},
\end{equation}
where $\Bbb{NN}$ denotes a set of functions described in a special form of neural network.

An important question is ``\textit{what is the minimum ratio of undersampling to provide guarantee of accurate reconstruction $f$ in \eqref{f-Mani-min}}?". It is closely related to the dimension of the manifold $\mathcal{M}_{\mbox{\tiny image}}$ and the capability of finding the generator $\mathfrak{G}$ in [H2].  Currently, our explicit knowledge on the solution prior (i.e., $\mathcal{M}_{\mbox{\tiny image}}$) is very limited and hardly built.

To clarify a concept of manifold prior, we try to solve and analyze underdetermined problems subjected to the model manifold, which is well-understood in a mathematical framework.

\subsection{A special solution manifold: Comparison of conventional methods with deep learning method}\label{M-example}
\begin{figure*}[t!]
\centering
\includegraphics[width=1\textwidth]{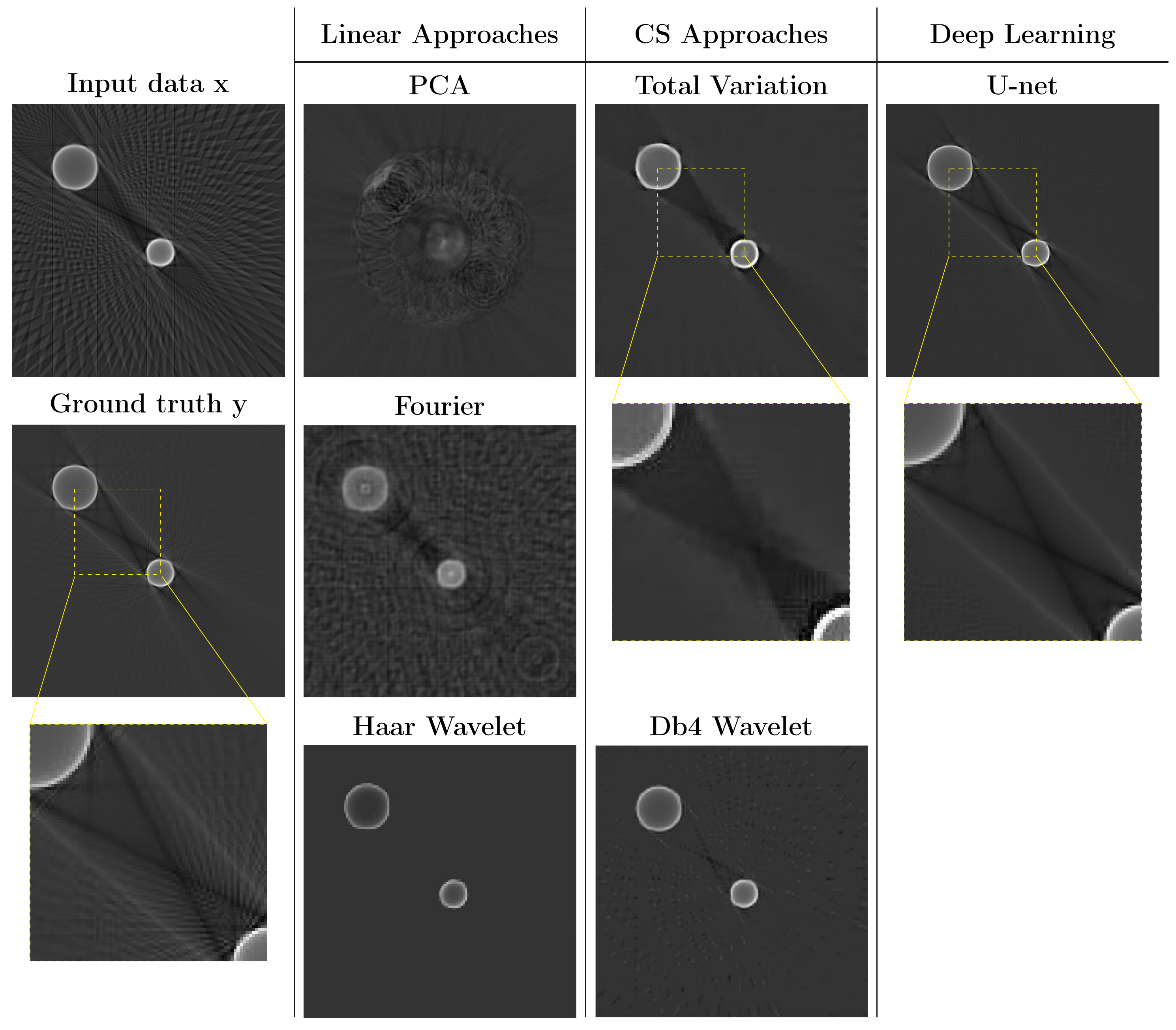}
\caption{Empirical results using various reconstruction approaches for the sparse-view CT problem with the special solution manifold in \eqref{Mani1}. For the linear projection approach, three different linear representations (PCA, Fourier, and Haar wavelet) of the input data were truncated at the 800th term after arranging the terms in the descending order according to the absolute value of their coefficients. For implementing the CS approach, we applied the $\ell_1$ convex relaxation method using two different transforms (total variation(TV) and Daubechies four tap(Db4) wavelet). An optimal regularization parameter was empirically selected between 0.01 and 1. Lastly, as the deep learning technique, U-net, trained by 800 training data pairs, is used.}
\label{FigMPCT}
\end{figure*}
This section provides a novel example of a low-dimensional manifold $\mathcal{M}_{\mbox{\tiny image}}$ to explain [H1]--[H3]. Using this manifold, we examine the capability of solving the sparse-view CT model using various exiting methods such as linear approaches (e.g. PCA, truncated Fourier and wavelet transform), sparse sensing (e.g. TV and Dictionary learning), and deep learning (e.g. U-net). This special solution manifold $\mathcal{M}_{\mbox{\tiny image}}$ highlights the advantage of deep learning method over conventional methods, where two approaches use prior information of the solution in a completely different way.

Our example of the manifold $\mathcal{M}_{\mbox{\tiny image}}$ in [H1] is seven dimensional and given by
\begin{equation} \label{Mani1}
\mathcal{M}_{\mbox{\tiny image}} := \{ ~ \mathfrak{G}(\h) \in \mathbb{R}^{n} ~ : ~ {\textbf{h}} \in \mathcal K \}
\end{equation}
where $n=256^2$, $\mathcal K$ is a compact subset of $\mathbb{R}^{7}$, and the continuous version of $\mathfrak{G} (\h)$ is given by
\begin{equation}\label{Mani2}
\mathfrak{G}(\h)=-\f{1}{4\pi}\mR^* \mathcal I^{{\tiny{-1}}}\left[\ln\left(\f{\sinh\left( h_7\mR\chi_{D_\textbf{h}}\right)}{h_7 \mR\chi_{D_\textbf{h}}}\right)\right]
\end{equation}
where the notations are the following:
\begin{itemize}
\item $\mR^* $ is the dual of the Radon transform $\mathcal{R}$. (See Section \ref{sec-localCT} for details.)
\item $\mathcal I^{{\tiny{-1}}}$ is the Riesz potential of degree -1.
\item ${\textbf{h}}=(h_1, h_2, \cdots, h_7)$.
\item $D_\textbf{h}$ is a union of two disks with centers $(h_1,h_2), (h_2,h_3)$ and radii $h_5, h_6$.
\item $\chi_{D}$ is the characteristic function of  $D$.
\end{itemize}
This example originates from the paper \cite{Park2017}, where the image of $\mathfrak{G}(\h)$ represents metal artifacts of CT in the presence of metallic objects occupied in the region $D_\textbf{h}$. Fig. \ref{FigMA} shows images on the manifold $\mathcal{M}_{\mbox{\tiny image}}$.

Assuming that $\mathfrak{G}$ is known, consider the highly underdetermined problem  \eqref{undersampled} to find the following reconstruction map:
\begin{equation} \label{fmap}
f: \x \in \mathcal{M}_{\mbox{\tiny image}}^\prime \mapsto \y \in \mathcal{M}_{\mbox{\tiny image}} \mbox{ satisfying } \A^\sharp \A \y = \x
\end{equation}
where
\begin{equation} \label{ULMC}
\mathcal{M}_{\mbox{\tiny image}}^\prime := \{ ~ \A^\sharp \A \mathfrak{G}(\h) ~ : ~ {\textbf{h}} \in \mathcal K ~ \}
\end{equation}
If  $m$ (the number of equations) is greater than seven, it is possible to find $f$ and this $f$ can be obtained as follows:
\begin{equation} \label{SPUD}
f(\x) =\mathfrak{G}(\h), ~~~~ \h=\underset{\textbf{h} \in \mathcal K}{\mbox{argmin}} ~ \| \A^\sharp \A \mathfrak{G}(\h) - \x \|_{\ell^2}^2
\end{equation}
However, we do not know $\mathfrak{G}$ in practice.

In the remaining part of this section, we examine the capability of various methods for finding a reconstruction map $f:\x\mapsto \y$ in \eqref{fmap} using the sparse-view CT model described in Section \ref{sec-sparseCT}.  Let $\{(\x^{(k)},\y^{(k)})\}_{k=1}^{\mathfrak{n}_{\mbox{\tiny data}}}$ denote a training data set.

\begin{figure*}[t!]
\centering
\includegraphics[width=0.9\textwidth]{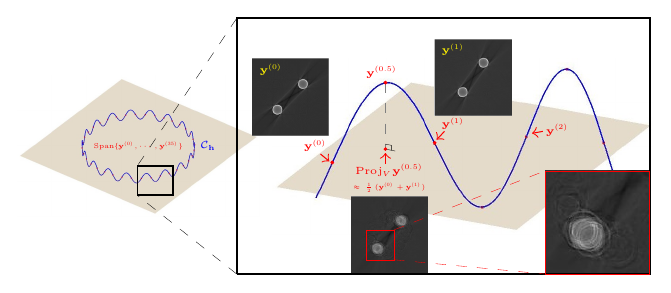}
\caption{Linear dimensionality reduction methods such as PCA may provide a poor approximation of the highly curved manifold in \eqref{Curve}. Let $V = \mbox{span}\{y^{(k)}\}_{k=0}^{35}$ 
	where $y^{(k)}$ is the $k\pi/18$ degree rotated image of the image $\y^{(0)}$. Let $y^{(0.5)}$ be the $\pi/36$ degree rotated image of the image $\y^{(0)}$. The projection of $y^{(0.5)}$ onto $V$ is approximately equal to $\frac{1}{2} (y^{(0)} + y^{(1)})$ that destroys the main characteristics of $y^{(0.5)}$.}
\label{FigPCA}
\end{figure*}
\subsubsection{Linear projection approach}
This subsection explains that there may not exist an appropriate low-dimensional linear projection that captures the variations in $\mathfrak{G}(\h)$.
Principal component analysis (PCA) is widely used for the dimensionality reduction in which the unknown manifold $\mathcal{M}_{\mbox{\tiny image}}$ is approximated by a linear subspace spanned by the set of principal components $\{\d_k\}_{k=1}^{\mathfrak{n}_{\mbox{\tiny PCA-basis}}}$. To be precise, the first principal component, $\d_1$, is obtained as follows:
\begin{equation}
\d_1 = \underset{|\textbf{d}|=1}{\mbox{argmax}} ~ \textbf{d}^T \textbf{Y}^T \textbf{Y} \textbf{d}
\end{equation}
where $\textbf{Y} := (\textbf{y}_1,\textbf{y}_2,\cdots,\textbf{y}_{\mathfrak{n}_{\mbox{\tiny data}}} )^T$. Similarly, the second principal component, $\d_2$, is obtained by computing the first principal component of matrix $\textbf{Y}_1 := \textbf{Y} - \textbf{d}_{1}\textbf{d}_{1}^T$. Continuing this process, we obtain the orthogonal basis $\{\d_k\}_{k=1}^{\mathfrak{n}_{\mbox{\tiny PCA-basis}}}$.
Subsequently, the reconstruction map $f$ is given by
\begin{equation} \label{SPUD2}
f(\x) =\textbf{D}\h, ~~~~ \h=\underset{\textbf{h}}{\mbox{argmin}} ~ \| \A^\sharp \A \textbf{D} \textbf{h} - \x \|_{\ell^2}^2
\end{equation}
where $\D$ denotes the matrix whose columns are $\{\d_k\}_{k=1}^{\mathfrak{n}_{\mbox{\tiny PCA-basis}}}$.

Fig. \ref{FigPCA} depicts that PCA fails to provide satisfactory approximations of images in the unknown manifold $\mathcal{M}_{\mbox{\tiny image}}$, because the low dimensional subspace spanned by the principal components cannot sufficiently cover the nonlinearity of the solution manifold. In Fig. \ref{FigPCA}, $\mathcal{C}_{\h}$ represents the following one-dimensional curve lying on the manifold $\mathcal{M}_{\mbox{\tiny image}}$:
\begin{equation} \label{Curve}
\mathcal{C}_{\h} := \{ \mathfrak{G}(\mathcal{T}_\theta \h) :  0<\theta\le 2\pi\}, \mathcal{T}_\theta=\begin{pmatrix}
\mathcal{R}_\theta & 0  & 0 \\ 0 & \mathcal{R}_\theta & 0\\
0&0 & {\bf I}
\end{pmatrix}
\end{equation}{\tiny }
where $\mathcal{R}_\theta$ is the rotation matrix with angle $\theta$, $ {\bf I}$ is $3\times 3$ identity matrix, and $0$ here denotes the corresponding zero matrix. The plane in Fig. \ref{FigPCA} represents the 36-dimensional space spanned by the sampled images $\{\y^{(k)}\}_{k=0}^{35}$, which are sampled at $\theta=\f{k}{36}2\pi,~k=0,\cdots,35$, on the curve $\mathcal{C}_{\h}$. Although $\mathcal{C}_{\h}$ is the map of the simple circle $\{ \mathcal{T}_\theta \h : 0 < \theta \leq 2\pi \}$ (in the latent space) through the generator function $\mathfrak{G}$, it is highly curved and complex due to the severe nonlinearity of function $\mathfrak{G}$. Therefore, with the limited expressivity of PCA \cite{Poole2016}, one cannot adequately approximate the curve $\mathcal{C}_{\h}$ by using the plane spanned by $\{\y^{(k)}\}_{k=0}^{35}$.

Fig. \ref{FigMPCT} depict that linear approaches, including PCA, in solving the underdetermined problem \eqref{undersampled} result in the significant loss of information from the original image. The inability of the linear projection approach to provide the global approximation of the highly curved image manifold is the reason for such poor reconstruction results.

\begin{figure*}[t!]
\centering
\includegraphics[width=0.95\textwidth]{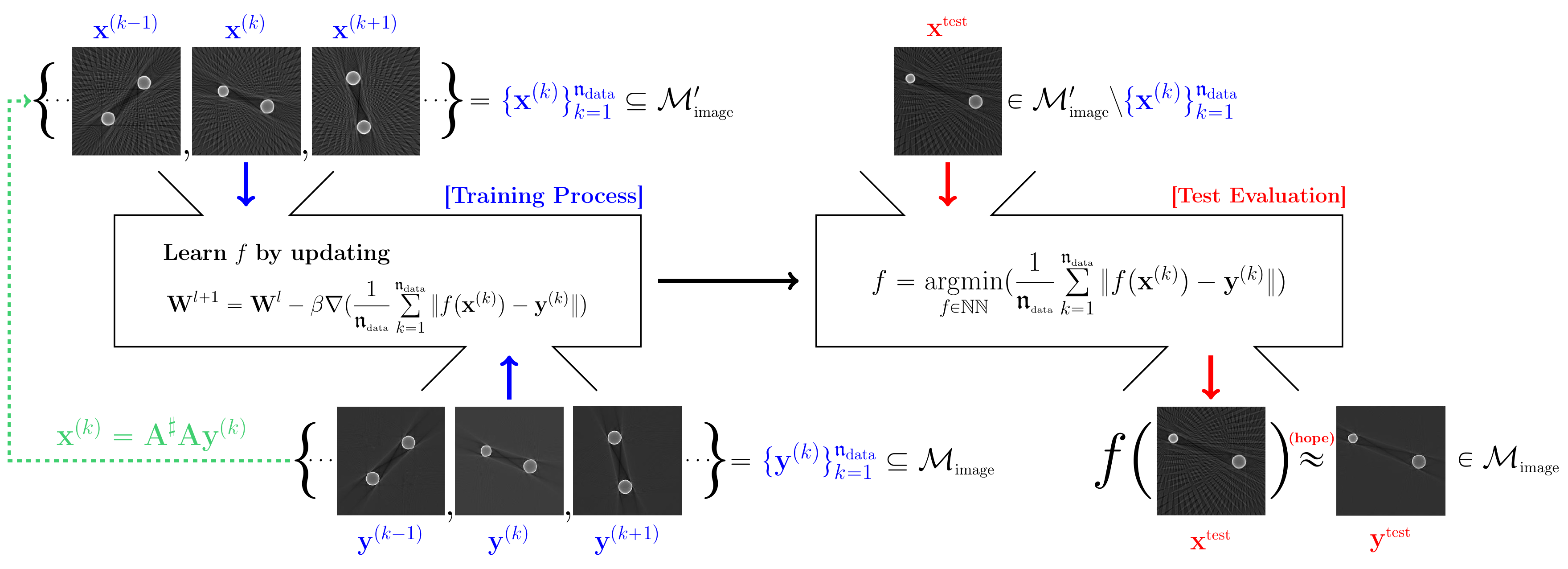}
\caption{Deep learning framework for estimating the reconstruction map $f : \x \mapsto \y$. For a given deep learning network and training dataset $\{\x^{(k)},\y^{(k)}\}_{k=1}^{\mathfrak{n}_{\mbox{\tiny data}}}$, the reconstruction map $f$ is obtained by minimizing the discrepancy between the network outputs $\{f(\x^{(k)})\} _{k=1}^{\mathfrak{n}_{\mbox{\tiny data}}}$ and the corresponding labels $\{\y^{(k)}\} _{k=1}^{\mathfrak{n}_{\mbox{\tiny data}}}$. Despite the finite number of training data, we hope the function $f$ to provide an accurate approximation of the test label $\y^{\mbox{\tiny test}}$ for any unobserved test data $\x^{\mbox{\tiny test}} \in \mathcal{M}_{\mbox{\tiny image}}^\prime \setminus \{\x^{(k)}\}_{k=1}^{\mathfrak{n}_{\mbox{\tiny data}}}$.}
\label{FigEPDL}
\end{figure*}	
\subsubsection{Compressed sensing approach} \label{CSsection}
Compressed sensing(CS) is based on the assumption that  $\y\in \mathcal{M}_{\mbox{\tiny image}}$ has sparse representation under a basis $\{\d_k\}_{k=1}^{\mathfrak{n}_{\mbox{\tiny CS-basis}}}$, i.e.,
\begin{equation} \label{sparse1}
\y = \textbf{D} \textbf{h} ~\ \mbox{   s.t.  } ~\| \textbf{h}\|_{\ell_0} \ll \mathfrak{n}_{\mbox{\tiny CS-basis}}
\end{equation}
where $\textbf{D}$ is a matrix whose $k$-th column corresponds to $\d_k$ and $\|\h\|_{\ell_0}$ is the number of non-zero entries of $\h$. In CS, $\ell^1$ convex relaxation methods are widely used to make the problem computationally feasible. A sparse approximation to the solution of the underdetermined problem \eqref{undersampled} is obtained as follows:
\begin{equation} \label{sparse3}
f(\x) =\textbf{D}\h, ~~\h=\underset{\textbf{h}}{\mbox{argmin}} ~ \| \A^\sharp \A \textbf{D} \textbf{h} - \x \|_{\ell^2}^2 + \lambda \| \h \|_{\ell^1}
\end{equation}
where $\lambda$ is a regularization parameter that controls the trade-off between data fidelity and the regularity enforcing the sparsity of $\h$. Kindly refer to \cite{Donoho2003,Donoho2006,Candes2006b,Candes2008,Bruckstein2009} for additional details. We implement the CS technique by using several wavelet bases, which are efficient in CS applications for natural images \cite{Daubechies2004,Mallat2009}. However, the reconstruction results from Fig. \ref{FigMPCT} show that some details are not preserved in the CS process.

The total variation(TV)-based CS method imposes a sparsity of the image gradient, where $f(\x)$ can be obtained as follows:
\begin{equation}\label{TV-min}
f(\x)=\underset{\textbf{y}}{\mbox{argmin}} ~ \| \A^\sharp \A \y - \x \|_{\ell^2}^2 + \lambda \| \nabla \y \|_{\ell^1}
\end{equation}
Fig. \ref{FigMPCT} shows that TV-based CS method also eliminates some of the details. TV method does not selectively preserve the streaking feature lying between two disks, while removing the other artifacts.

Dictionary learning \cite{Olshausen1996,Aharon2006} utilizes the given training data to find a (redundant and data-driven) basis $\{\d_k\}_{k=1}^{\mathfrak{n}_{\mbox{\tiny dic}}}$ that can represent every $\y \in \mathcal{M}_{\mbox{\tiny image}}$ as a sparse vector. A learned dictionary can handle a specific problem considerably better than analysis-driven dictionaries (e.g. wavelet and framelet) \cite{Elad2006,Aharon2008,Mairal2008,Yang2010,Rubinstein2010}. Dictionary learning approaches have a drawback in dealing with high-dimensional data due to the huge computational complexity; hence, the patch-based approach (e.g. image patch of size $8\times 8$ pixels) has been adopted in most image processing applications. However, this approach might not be fit for the tasks for which the global information should be sufficiently incorporated.

\begin{figure*}[t!]
\centering
\includegraphics[width=0.95\textwidth]{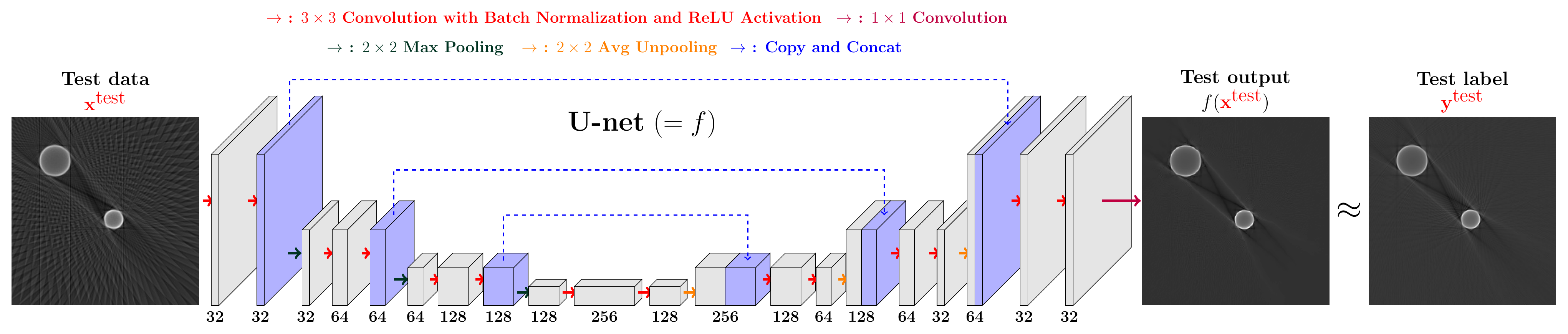}
\caption{Deep learning approach using U-net for solving the sparse-view CT model with the special solution manifold. By training U-net with 800 data pairs, the reconstruction function $f$ is obtained. In Tensorflow environment, the minimization process was performed by using Batch normalization and Adam optimizer with learning rate 0.001, mini-batch size 16, and 1000 epochs. For a given test data $\x^{\mbox{\tiny test}}$, the reconstruction function $f$ provides the test output $f(\x^{\mbox{\tiny test}})$ which approximates the test label $\y^{\mbox{\tiny test}}$.}
\label{FigMPDL}
\end{figure*}
\subsubsection{Deep learning approach}
Deep learning techniques expand our ability to solve underdetermined problems via sophisticated learning process by using group-data fidelity of the training data; furthermore, they appear to effectively deal with the limitations of the existing mathematical methods in handling various ill-posed problems. In CS, $f(\x)$ in \eqref{TV-min} can be viewed as a solution of the nonlinear Euler-Lagrange equation associated with the trade-off between two separative competitive objectives of maximizing the ``single data fidelity" and minimizing TV (as a sparse prior of natural images). However, this sparse prior may not be appropriate for preserving small features that contain clinically useful information. In contrast, the deep learning approach \eqref{MDL} utilizes ``group-data fidelity" to estimate the reconstruction map $f : \x \mapsto \y$ by seemingly probing the relationship between unknown manifolds $\mathcal M_{\mbox{\tiny image}}$ and $\mathcal M_{\mbox{\tiny image}}'$.  The reconstruction $f$ is obtained by minimizing the group-data discrepancy $\sum_{k=1}^{\mathfrak{n}_{\mbox{\tiny data}}}\| f(\x^{(k)}) - \y^{(k)} \|$ (i.e. maximizing the group-data fidelity) on a finite number of training pairs $\{(\x^{(k)}, \y^{(k)})\}_{k=1}^{\mathfrak{n}_{\mbox{\tiny data}}}$ lying on $\mathcal M_{\mbox{\tiny image}}^\prime \times \mathcal M_{\mbox{\tiny image}}$, as shown in Fig. \ref{FigEPDL}.

In particular, U-net \cite{Ronneberger2015} has achieved enormous success in finding the map for various underdetermined medical imaging problems \cite{Hyun2018,Jin2017,Han2017}. In U-net, the network architecture of $f$ comprises a contraction path $\Phi : \x \mapsto \h$ and an expansion path $\Psi : \h \mapsto \y$; $f(\x) = \Psi \circ \Phi(\x)$. To be precise, the simplest form of the contraction path $\Phi$ is expressed by
\begin{equation}
\h = \Phi(\x) = \sigma(\textbf{w}_3 \circledast \mathcal P(\sigma(\textbf{w}_2 \circledast \sigma(\textbf{w}_1 \circledast \x + \c_1 )  + \c_2 ))+\c_3)
\end{equation}
and the corresponding expansive path $\Psi$ is represented as
\begin{equation}
\Psi(\h) = \textbf{w}_6 \odot ( \sigma( \textbf{w}_5 \circledast (\mathcal C_{cat} ( \mathcal U_{pool} ( \sigma(\textbf{w}_4 \circledast \h + \c_4 ) ),  \z ) ) + \c_5 ) )
\end{equation}
where $\z = \sigma(\textbf{w}_2 \circledast \sigma(\textbf{w }_1 \circledast \x + \c_1) + \c_2)$. Here, $\sigma(z) = \mbox{ReLU}(z)$, $\mathcal P$ is a pooling operator, $\mathcal U_{pool}$ is an unpooling operator, and $\mathcal C_{cat}$ is a concatenation operator. The work in \cite{Ronneberger2015} can be referred for a more detailed description. The overall structure of U-net is shown in Fig. \ref{FigMPDL}.

The map $f : \x \mapsto \y$, as a function of parameters $\textbf{W} = \{ \textbf{w}_1, \textbf{c}_1, \textbf{w}_2, \textbf{c}_2, \cdots \}$, is determined as follows:
\begin{equation} \label{DLEQ}
f = \underset{f \in \mathbb{N}\mathbb{N}}{\mbox{argmin}} ~ \dfrac{1}{\mathfrak{n}_{\mbox{\tiny data}}}\sum_{k=1}^{\mathfrak{n}_{\mbox{\tiny data}}}\| f(\x^{(k)}) - \y^{(k)} \|_{\ell^2}^2
\end{equation}
where $\mathbb{N}\mathbb{N}$ denotes a set of all the functions of the form $f = \Psi \circ \Phi$ that vary with $\textbf{W}$.  Fig. \ref{FigMPCT} shows remarkable performance of U-net.

\section{Solvability of Underdetermined Linear System} \label{sec-Solvability}
\begin{figure*}[t!]
\begin{center}
	\includegraphics[width=0.8\textwidth]{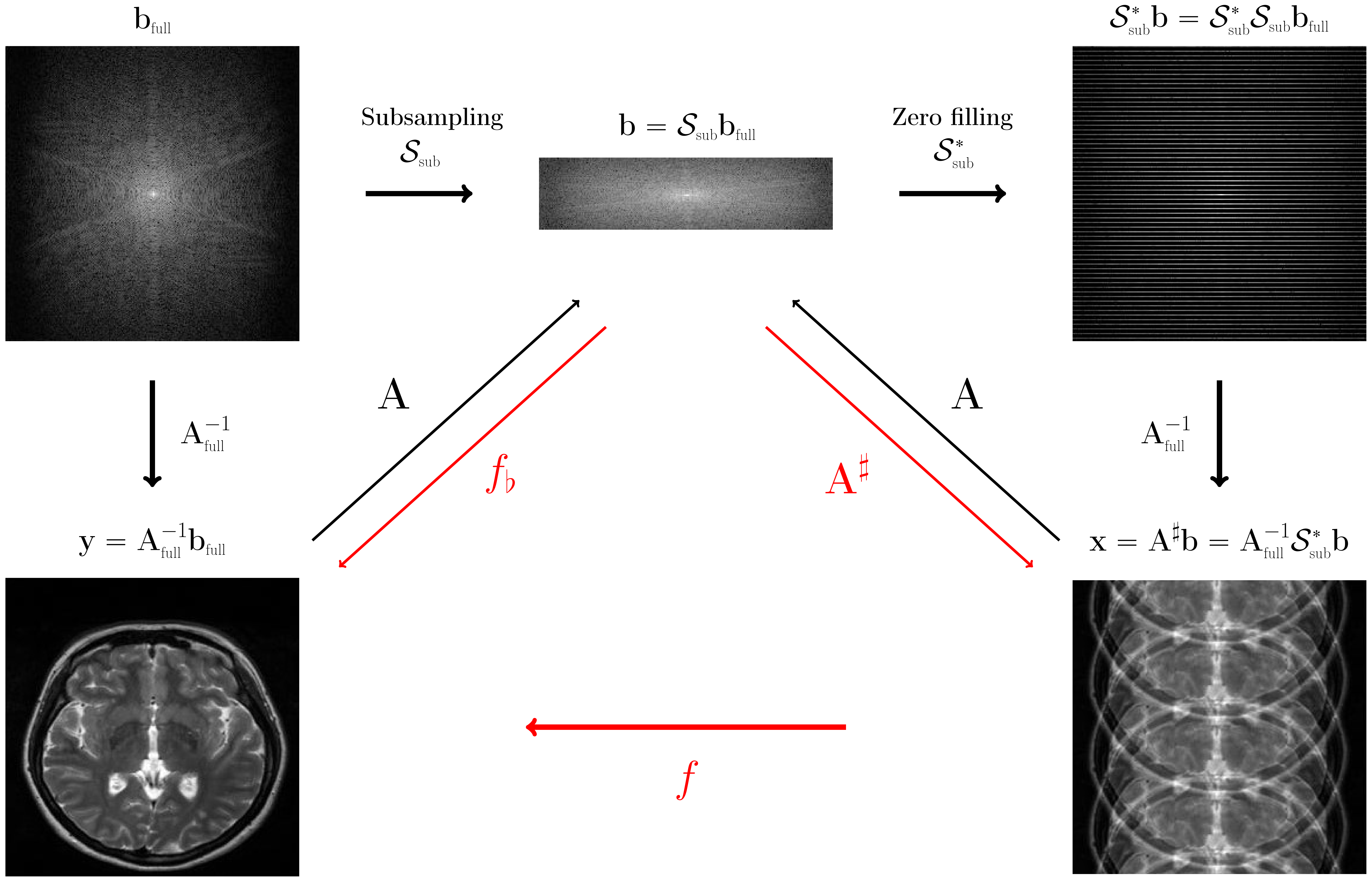}
	\caption{Undersampled MRI problem is to recover an image $\y=\A^{-1}_{\mbox{\tiny full}}\b_{\mbox{\tiny full}}$ from undersampled data $\b=\mathcal S_{\mbox{\tiny sub}} \b_{\mbox{\tiny full}}$, where $\mathcal S_{\mbox{\tiny sub}}$ is a subsampling operator and $\b_{\mbox{\tiny full}}$ is a fully-sampled data in the sense of Nyquist sampling. All images are displayed by taking their absolute values. Using the deep learning technique, we attempt to find a reconstruction function $f$ that maps from $\x = \A^\sharp \b$ to $\y$. Since the structure of $\x$ is determined by the subsampling operator $\mathcal S_{\mbox{\tiny sub}}$, the learning $f$ can be affected by the subsampling strategy.}
	\label{Fig-under-MRI1}
\end{center}
\end{figure*}

In undersampled problems, the subsampling strategy $\mathcal{S}_{\mbox{\tiny sub}}$ inside $\A = \mathcal{S}_{\mbox{\tiny sub}}\A_{\mbox{\tiny full}}$ is important for the uniqueness of solution $\y$ on the manifold $\mathcal{M}_{\mbox{\tiny image}}$ among all the possible solutions in $\mathcal N_{\textbf{b}}(\A)$. Precisely, a proper subsampling strategy $\mathcal{S}_{\mbox{\tiny sub}}$ is related to the following manifold restricted isometry property (RIP) condition. The matrix $\A$ associated with  $\mathcal{S}_{\mbox{\tiny sub}}$ is said to satisfy the $\mathcal{M}$-RIP condition if there exists a constant $c \in (0,1]$ such that
\begin{equation}\label{MRIP}
c \| \y-\y'\| \le  \| \A \y - \A \y'\|\le \f{1}{c} \|\y-\y'\| ~\mbox{for all } \y,\y' \in \mathcal{M}_{\mbox{\tiny image}}.
\end{equation}

The following two observations explain the necessary condition for constructing a suitable subsampling strategy:
\begin{obs}\label{lemOTO}
If $\A$ satisfies the $\mathcal{M}$-RIP condition in \eqref{MRIP}, then
\begin{equation}\label{A-sampling}
\A^\sharp \A : \mathcal{M}_{\mbox{\tiny image}} \mapsto \mathcal{M}^\prime_{\mbox{\tiny image}}~~ \mbox{is one-to-one.}
\end{equation}	
\end{obs}
\begin{proof}
Suppose that there are two different $\y$ and $\y^\prime$ such that $\A^\sharp \A \y = \A^\sharp \A \y^\prime$. Since $\A^\sharp = \A_{\mbox{\tiny full}}^{-1} \mathcal S_{\mbox{\tiny sub}}^*$,
$$
0 = \| \A_{\mbox{\tiny full}}(\A^\sharp \A \y - \A^\sharp \A \y^\prime) \| = \| \mathcal S_{\mbox{\tiny sub}}^*(\A\y-\A\y^\prime) \| \nonumber
$$
$$
= \| \A\y - \A\y^\prime \| \geq c \| \y - \y^\prime \|
$$
where the last inequality follows from \eqref{MRIP}. Hence, $\y - \y^\prime =0$, which contradicts the assumption.
\end{proof}

\begin{figure*}[t!]
\begin{center}
	\includegraphics[width=0.86\textwidth]{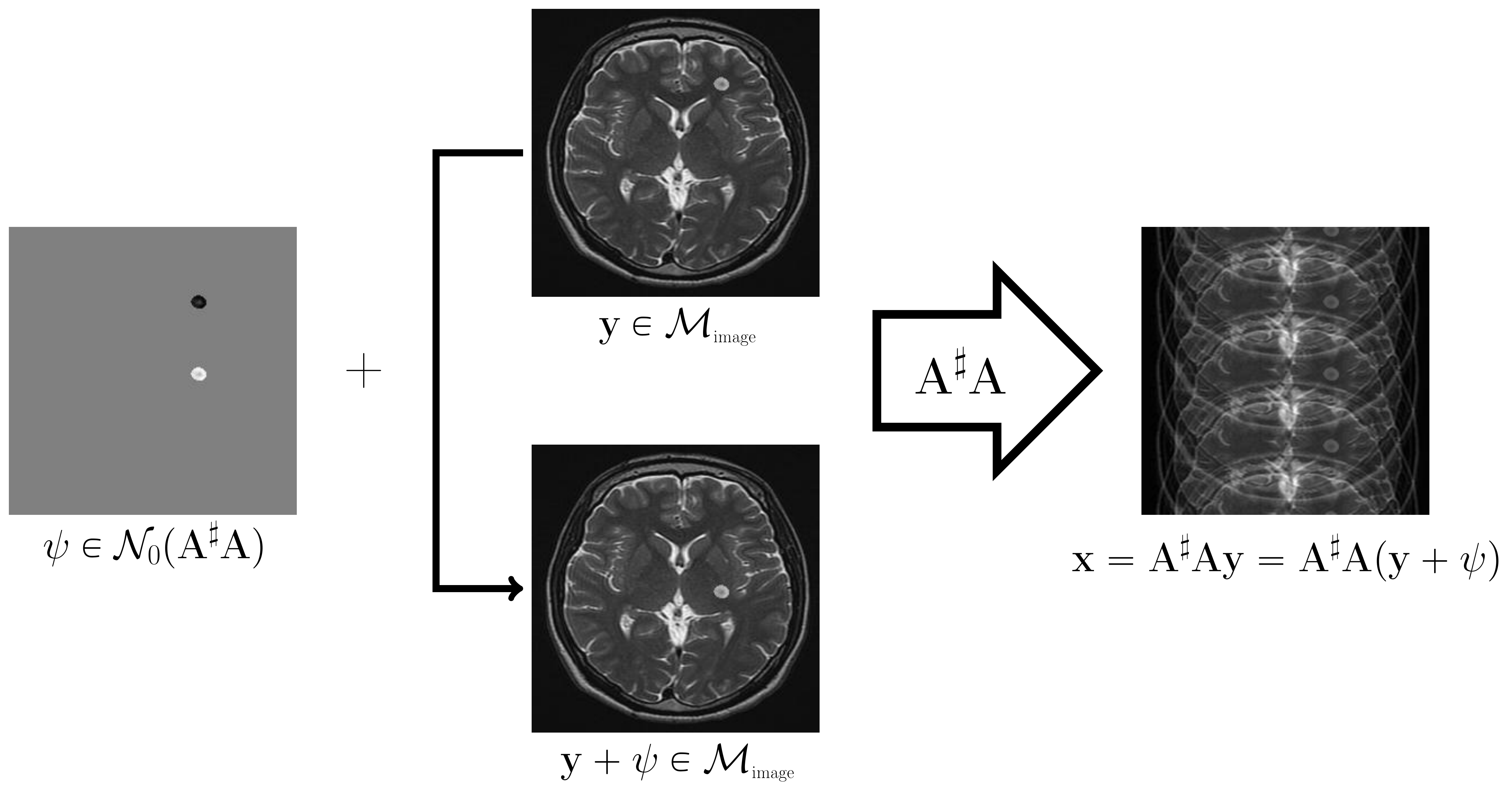}
	\caption{Location uncertainty on the solution manifold when using a uniform subsampling with factor 4; let us consider two different MR images, $\y$ and $\y + \psi$, where the location of a small anomaly is only different. When taking $\textbf{A}^{\sharp}\textbf{A}$ to the images, we obtain the same output $\x = \textbf{A}^{\sharp}\textbf{A}\y = \textbf{A}^{\sharp}\textbf{A}(\y + \psi)$, where $\textbf{A}$ is given by $\textbf{A}=\mathcal S_{\mbox{\tiny sub}}\textbf{A}_{\mbox{\tiny full}}$ and $\mathcal S_{\mbox{\tiny sub}}$ denotes a uniform subsampling with factor 4.}
	\label{UniformSamplingreal}
\end{center}
\end{figure*}

\begin{obs}\label{obsMRIP}
The reconstruction map $f : \x\in\mathcal{M}^\prime_{\mbox{\tiny image}}  \mapsto \y\in \mathcal{M}_{\mbox{\tiny image}}$ is learnable if $\A$ satisfies the $\mathcal{M}$-RIP condition.
\end{obs}
If the corresponding matrix $\A$ does not satisfy the $\mathcal{M}$-RIP condition, there exist $\y_1 \neq \y_2$ such that $\x=\A^\sharp \A\y_1= \A^\sharp \A\y_2$; therefore, it is impossible to learn such $f$ due to indistinguishability. The issue of learnability associated with $\A$ that does not satisfy the $\mathcal{M}$-RIP condition will be addressed in Section \ref{SolveUMRI} with a concrete example.

Given a highly undersampling operator $\mathcal{S}_{\mbox{\tiny sub}}$, the map $f$ can be viewed as an image restoration function with filling-in missing data or unfolding image data; therefore, $f(\x)$ depends on the image structure. The nonlinearity of $f$ is affected by $\mathcal{S}_{\mbox{\tiny sub}}$ and the degree of bending of the manifold $\mathcal{M}_{\mbox{\tiny image}}$. The following observation explains that most problems of solving underdetermined linear systems in medical imaging are highly non-linear.
\begin{obs} Suppose that $\A$ satisfies the $\mathcal{M}$-RIP condition.  Let ${\bf V}_{\mathcal{M}_{\mbox{\tiny image}}}$ be the span of the set $\{\f{\p}{\p h_j} \mathfrak{G}(\h): ~\h\in \mathcal K, j=1,\cdots, \mathfrak{d}_{\mbox{\tiny mfd}}\}$. If $\mbox{dim} ~ {\bf V}_{\mathcal{M}_{\mbox{\tiny image}}}  >  m$, then the reconstruction map $f : \x\in\mathcal{M}^\prime_{\mbox{\tiny image}}  \mapsto \y\in \mathcal{M}_{\mbox{\tiny image}}$ is non-linear.
\end{obs}
\begin{proof}
Note that $f$ satisfies  $f(\x)=\y$ with $\x=\A^\sharp\A \y$ for all $\y \in \mathcal{M}_{\mbox{\tiny image}}$.
Since $\mathcal M_{\mbox{\tiny image}}$ is generated by $\mathfrak{G}$, we obtain
\begin{equation}
f(\A^\sharp\A \mathfrak{G}(\h)) = \mathfrak{G}(\h), \q \forall ~ \h \in \mathcal{K}
\end{equation}
Taking gradient with respect to $\h$ on both sides, then
\begin{equation}\label{AAG}
\nabla_{\x} f(\A^\sharp\A \mathfrak{G}(\h)) \A^\sharp \A \nabla_{\h} \mathfrak{G}(\h) = \nabla_{\h}  \mathfrak{G}(\h),\forall \h \in \mathcal{K}
\end{equation}
To derive a contradiction, suppose $f$ is linear; i.e., there exists a fixed matrix $\B \in \mathbb{R}^{n\times n}$ such that  $\nabla_{\x} f (\x) = \textbf{B}$ for all $\x\in \A^\sharp\A\mathcal{M}_{\mbox{\tiny image}}$. Subsequently, \eqref{AAG} becomes
\begin{equation}
\B \A^\sharp \A \nabla \mathfrak{G}(\h) = \nabla \mathfrak{G}(\h), \q \forall ~ \h \in \mathcal{K}
\end{equation}
Hence, denoting the eigenspace of $\textbf{B}\textbf{A}^\sharp \textbf{A}$ corresponding to the eigenvalue $\lambda$ by $E_{\lambda} (\B \A^\sharp \A)$, we have
\begin{equation}
E_{1} (\B \A^\sharp \A) \supseteq {\bf V}_{\mathcal{M}_{\mbox{\tiny image}}}
\end{equation}
and from the assumption on the dimension of ${\bf V}_{\mathcal{M}_{\mbox{\tiny image}}}$,
\begin{equation}
\mbox{dim} ~ E_{1} (\B \A^\sharp \A) > m.
\end{equation}
Since $\mbox{Rank}(\B\A^\sharp \A)\le \mbox{Rank}(\A) = m$,
\begin{equation}
\mbox{dim} ~ E_{0} (\B \A^\sharp \A) \ge n - m.
\end{equation}
This is a contradiction because
\begin{equation}
\mbox{dim} ~ E_{0} (\B \A^\sharp \A)+ \mbox{dim} ~ E_{1} (\B \A^\sharp \A) > n.
\end{equation}
\end{proof}

\begin{figure*}[t!]
\begin{center}
	\includegraphics[width=0.9\textwidth]{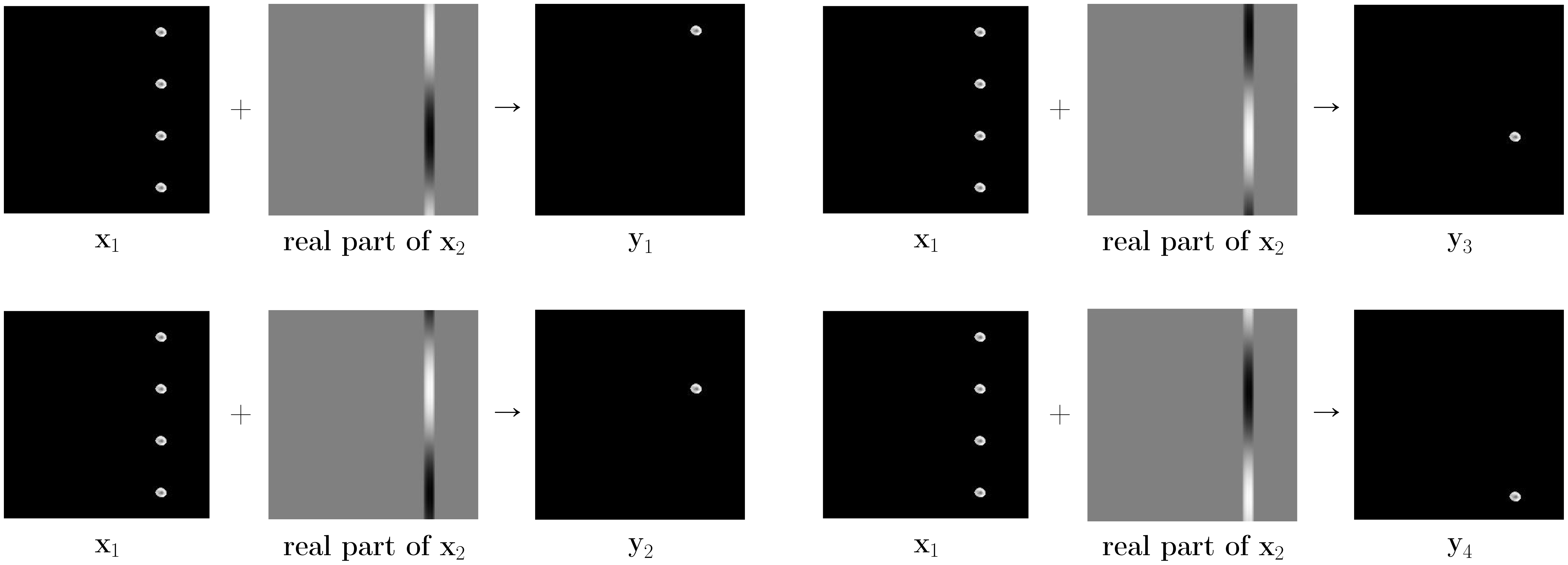}
	\caption{Empirical observation regarding how to eliminate the location uncertainty by adding one phase encoding line in the $k$-space. Four images ($\y_1, \y_2, \y_3,$ and $\y_4$) containing one small anomaly in four different locations generate the same $\x_1$; however, one additional phase encoding line information $\x_2$ can deal with location uncertainty in $\x_1$.}
	\label{SolLU}
\end{center}
\end{figure*}

\subsection{Undersampled MRI} \label{SolveUMRI}
In MRI, we apply an oscillating magnetic field to the imaging object in an MR scanner (being confined in a strong magnetic field) to acquire the $k$-space data ($\b$), which is used to produce a cross-sectional MR image $\y$.
In fully sampled MRI, the relation between a 2D MR image $\y$ and the corresponding fully sampled k-space data
$\b_{\mbox{\tiny{full}}}$ can be expressed in the following form \cite{Seo2014}:
\begin{equation}\label{MR-data-0}
\underbrace{\sum_{a,b=1,\cdots, \sqrt{n}}  ~e^{-2\pi i(a k_1 \Delta k +b k_2 \Delta k)} \y(a,b)}_{\A_{\mbox{\tiny{full}}}\y} ~=~ \b_{\mbox{\tiny{full}}}(k_1, k_2)
\end{equation}
where  $\Delta k$ denotes the Nyquist sampling distance, which is chosen in such a way that
\begin{equation}\label{MR-data-1}
\y(a,b)=\sum_{k_1,k_1=1,\cdots, \sqrt{n}}  ~e^{2\pi i(a k_1 \Delta k +b k_2 \Delta k)} \b_{\mbox{\tiny{full}}}(k_1, k_2).
\end{equation}
In other words, the Nyquist sampling make the problem $\A_{\mbox{\tiny{full}}}\y = \b_{\mbox{\tiny{full}}}$ well-posed so that the standard reconstruction $\y = \A_{\mbox{\tiny{full}}}^{-1} \b_{\mbox{\tiny{full}}}$ can be obtained by 2D discrete inverse Fourier transform.

Assume that the frequency-encoding is along the $k_1$-axis and that the phase-encoding is along the $k_2$-axis in the $k$-space. Noting that the MRI scan time is roughly proportional to the number of time consuming phase-encoding steps in $k$-space, there have been numerous attempts to shorten the MRI scan time by skipping the phase-encoding lines in the $k$-space \cite{Sodickson1997,Haacke1999}.  In the undersampled MRI, we attempt to find the optimal reconstruction function that maps the highly undersampled $k$-space data ($\b$ that violates Nyquist sampling criterion) to an image ($\y$) close to the MR image corresponding to the fully sampled data ($\b_{\mbox{\tiny{full}}}$ that satisfies the Nyquist sampling criterion).

With undersampled data $\b$, the corresponding problem is
\begin{equation}\label{Axb-MRI1}
\underbrace{\mathcal{S}_{\mbox{\tiny sub}}\A_{\mbox{\tiny full}}}_{\A}\u = \underbrace{\mathcal{S}_{\mbox{\tiny sub}}(\b_{\mbox{\tiny full}})}_{\b}
\end{equation}
where $\mathcal{S}_{\mbox{\tiny sub}}$ denotes a subsampling operator and $\b=\mathcal{S}_{\mbox{\tiny sub}}(\b_{\mbox{\tiny full}})$. The image $\x=\A^\sharp\b$ is one of the solutions of \eqref{Axb-MRI1}, because $\A^\sharp$ is the pseudo-inverse of $\A$ in this case. The undersampled MRI problem aims to find an image restoration map $f : \x=\A_{\mbox{\tiny full}}^{-1} \mathcal{S}_{\mbox{\tiny sub}}^*\b \mapsto \y=\A_{\mbox{\tiny{full}}}^{-1} \b_{\mbox{\tiny{full}}}$. See Fig. \ref{Fig-under-MRI1}.

\subsubsection{Uniform subsampling}
According to the Poisson summation formula, the discrete Fourier transform of the uniformly subsampled data with factor 4 produces the following four-folded image \cite{Seo2013}:
\begin{equation}\label{Poisson0}
\x(a,b)=\A^\sharp\A \y=\dfrac{1}{4} \underset{b'\equiv b ~ (\mbox{\scriptsize mod}\f{\sqrt{n}}{4}) }{\sum} \y(a,b')
\end{equation}
where $b'\equiv b ~ (\mbox{\scriptsize mod}\f{\sqrt{n}}{4})$ means that both $b$ and $b'$ leave the same remainder when divided by $\f{\sqrt{n}}{4}$. Unfortunately, there exists an uncertainty that makes it impossible to reconstruct $\y$ from $\x$, and therefore $f$ is not learnable. To see the reason, we consider the following:
\begin{equation}\label{KernelAA}
\Psi_{\mbox{\tiny ufm}}:=\mathcal{N}_0(\A^\sharp\A) \nonumber =\mbox{Span}\{ \psi_{a_*,b_*}^{0,\beta}: a_*, b_*\in \Bbb Z_{\sqrt{n}},  \beta =1,2,3 \}
\end{equation}
where $\Bbb Z_{n}:=\{1, \cdots, n \}$ for any positive integer $n$ and $\psi_{a_*,b_*}^{0,\beta}$ is given by
\begin{equation}\label{Poisson2}
\psi_{a_*,b_*}^{0,\beta}(a,b)=
\left\{
\begin{array}{cl}
1 & \mbox{if } (a,b)=(a_*, b_*) \\
-1 & \mbox{if } (a,b)=(a_*, b_*)+(0, \f{\sqrt{n}}{4} \beta ) \\
0    & \mbox{otherwise}
\end{array}
\right.
\end{equation}
Here, $b_*+ \f{\sqrt{l}}{4} \beta$ should be understood as modulo $\sqrt{n}$.
\begin{obs} \label{ob3}
There exists a non-zero $\psi\in \Psi_{\mbox{\tiny ufm}}$  and $\y\in\mathcal M_{\mbox{\tiny image}}$ such that $\y+\psi\in\mathcal M_{\mbox{\tiny image}}$.
\end{obs}
The observation implies that the $\mathcal{M}$-RIP condition does not hold, as $f$ requires the following contradictory two conditions $f(\textbf{x}) = \y \mbox{ and } f(\textbf{x} ) = \y + \psi$, where $\textbf{x}=\textbf{A}^\sharp\textbf{A}\textbf{y} =\textbf{A}^\sharp\textbf{A}(\textbf{y}+\psi)$. The location of a small anomaly cannot be determined, and, therefore, there are many location uncertainties under the uniform subsampling, as shown in Fig. \ref{UniformSamplingreal}. This is the main reason why $f$ is not learnable under the uniform subsampling.

\begin{figure*}[t!]
\begin{center}
	\includegraphics[width=0.8\textwidth]{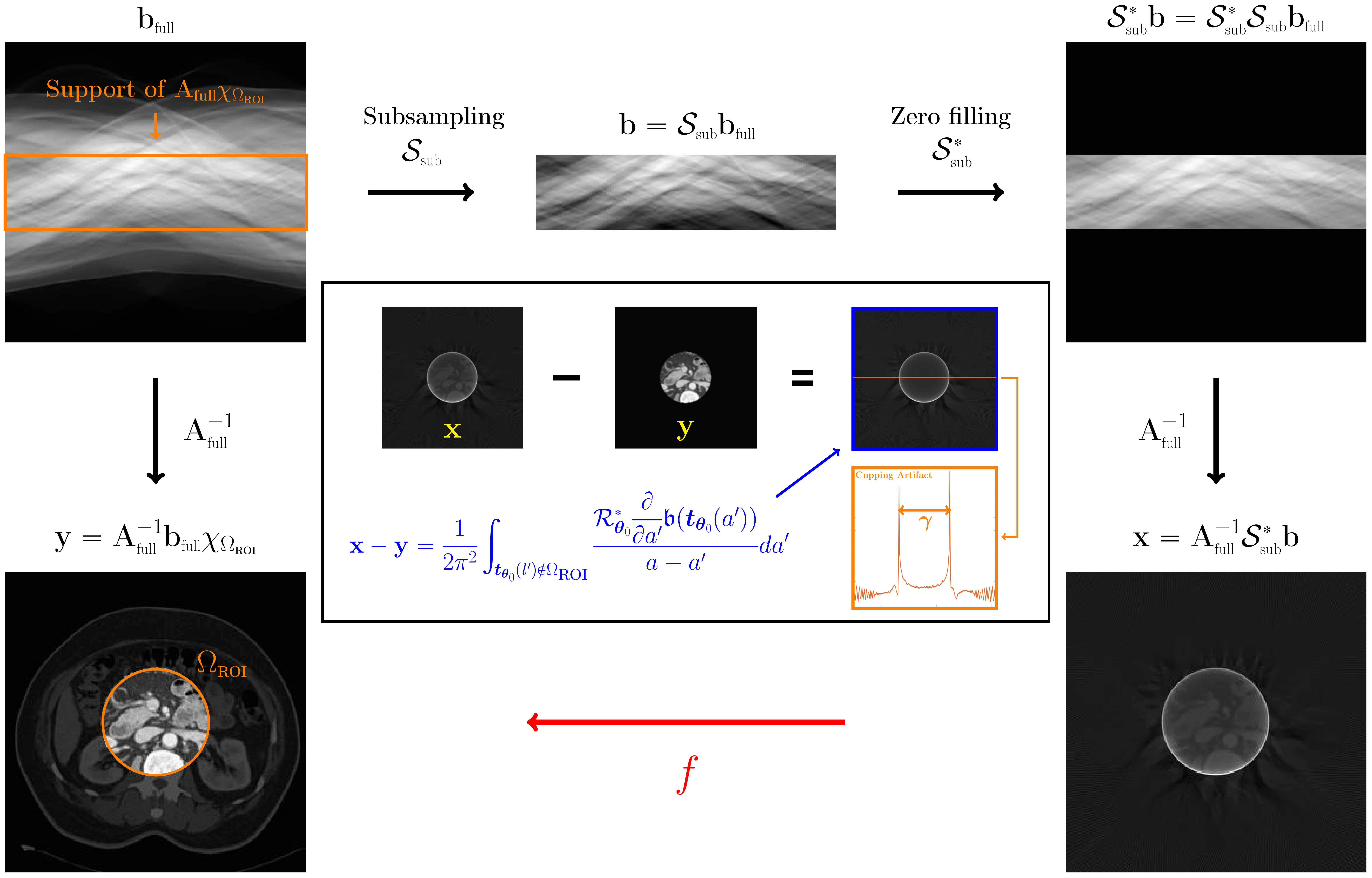}
	\caption{Interior tomography problem is to recover an image $\y = \A^{-1}_{\mbox{\tiny full}} \b_{\mbox{\tiny full}} \chi_{\Omega_{\mbox{\tiny ROI}}}$ in our region of interest (ROI) $\Omega_{\mbox{\tiny ROI}}$ by using the truncated data $\b = \mathcal{S}_{\mbox{\tiny sub}}\b_{\mbox{\tiny full}}$, where $\mathcal{S}_{\mbox{\tiny sub}}$ is a subsampling operator and $\chi$ is a characteristic function. Applying the deep learning method, the reconstruction map $f : \x \mapsto \y$ is learnable because of the analyticity of residual $\x-\y$.}
	\label{LT}
\end{center}
\end{figure*}

\subsubsection{Uniform sampling with adding one phase encoding line}
This section provides a way to improve the separability by adding only one phase encoding line to a uniform subsampling. Let $\mathcal S_{\mbox{\tiny{sub}}}$ be the uniform subsampling of factor 4 upon adding one phase encoding line. Then, $\x=\A^\sharp\b$ can be decomposed into two parts:
\begin{equation}\label{Poisson1}
\x(a,b)= \x_1(a,b) + \x_2(a,b)
\end{equation}
where $\x_1$ is the uniform sampling part given by
\begin{align}
\x_1(a,b):=\dfrac{1}{4}\underset{b'~\equiv b ~(\mbox{\scriptsize mod}\f{\sqrt{n}}{4}) }{\sum} \y(a,b')
\end{align}
and $\x_2$ is the single phase encoding part given by
\begin{align}
\x_2(a,b):=\underset{b'\in\Bbb Z_{\sqrt{n}}}{\sum}\y(a,b')e^{2\pi i (b-b') \Delta k}
\end{align}
Adding the additional low frequency line in the k-space (compared to the previous uniform sampling) provides the additional information of $\x_2$. Subsequently, the situation is dramatically changed to counter the anomaly-location uncertainty in uniform sampling. Fig. \ref{SolLU} shows why the $\x_2$ information can effectively handle the location uncertainty in Observation \ref{ob3}.

\subsection{Interior tomography}\label{sec-localCT}
This section explains the underdetermined system for the interior tomography problem. For simplicity, let us consider a 2-D parallel beam system and assume that the projection data for the entire field of view(FOV) is given by
\begin{equation} \label{LTCT}
\mathfrak{b}_{\mbox{\tiny{full}}}(\varphi,s)=\mR \mathfrak{u}(\varphi,s) := \int_{\R^2}\mathfrak{u}(\bt)\delta(\boldsymbol{\theta} \cdot\bt-s)d\bt
\end{equation}
where $\mathfrak{u}$ represents an attenuation distribution on 2D-slice, $\bt=(t_1,t_2)$, $\boldsymbol{\theta}=(\cos\varphi,\sin\varphi)$, and $\delta(\cdot)$ is the Dirac delta function. The discrete version of \eqref{LTCT} can be expressed by the following linear system
\begin{equation} \label{linearEq1}
\A_{\mbox{\tiny{full}}}\u =\b_{\mbox{\tiny{full}}}
\end{equation}
where the Nyquist criterion must be considered in terms of the expected resolution of the CT image($\u$) and sampling($\b$).
The standard reconstruction $\u=\A_{\mbox{\tiny{full}}}^{-1}\b_{\mbox{\tiny{full}}}$ is based on the FBP algorithm, which is based on the following identity:
\begin{equation} \label{FBP}
\mathfrak{u}(\bt) = \int_{0}^{\pi}\int_{\R} |\omega|\mathcal{F}\mathfrak{b}_{\mbox{\tiny{full}}}(\varphi,\omega) e^{2\pi i \omega \bt \cdot \boldsymbol{\theta} } d\omega d\varphi
\end{equation}
where $\mathcal{F}$ is 1 dimensional Fourier transform associated with the variable $\omega$ and $\mathfrak{u}$ and $\mathfrak{b}$ are the continuous forms of $\u$ and $\b_{\mbox{\tiny{full}}}$, respectively, in \eqref{linearEq1} \cite{Natterer1986}.

Now, we are ready to explain the interior tomography problem. Let $\roi\subset \R^2$ denote the local region of interests(ROI), whose size in the interior tomography is smaller than that of a patient's body to be scanned, as depicted in Fig. \ref{LT}. In the interior tomography, we attempt to reconstruct $\mathfrak{u}\chi_\roi$ by using the truncated data $\mathfrak{b}_{\mbox{\tiny{full}}}\chi_D$, where $\chi_\roi$ is the characteristic function of $\roi$ and $D$ is the support of $\mR\chi_\roi$.

The discrete form of this interior tomography can be expressed as follows:
\begin{equation}
\mbox{Reconstruct $\y = \u \chi_\roi$ satisfying $\A\u = \b$}
\end{equation}
where $\b=\mathcal{S}_{\mbox{\tiny sub}}(\b_{\mbox{\tiny full}})$, with the subsampling $\mathcal{S}_{\mbox{\tiny sub}}$ defined as $\mathcal{S}_{\mbox{\tiny sub}}(\b_{\mbox{\tiny full}}):=\b_{\mbox{\tiny full}}\chi_D$, and $\A := \mathcal{S}_{\mbox{\tiny sub}}\A_{\mbox{\tiny full}}$.
The dual operator of $\mathcal{S}_{\mbox{\tiny sub}}$, which is denoted by $\mathcal{S}_{\mbox{\tiny sub}}^*$, can be interpreted as a zero-filling process in the unmeasured parts of $\b$ in terms of $\b_{\mbox{\tiny{full}}}$, as shown in Fig \ref{LT}.
The standard FBP algorithm for the zero-filled data $\mathcal{S}_{\mbox{\tiny sub}}^*\b$ provides the image $\x=\A_{\mbox{\tiny full}}^{-1} \mathcal{S}_{\mbox{\tiny sub}}^*\textbf{b}$ with cupping artifacts \cite{Faridani1992,Wang2013}. Then, our reconstruction problem is the following:
\begin{align}\label{LCF}
& \mbox{Find a function $f:\x\mapsto \y$} \nonumber \\ & \mbox{satisfying  $\A^\sharp\A\u =\x$ and $\y = \u \chi_\roi$.}
\end{align}

To explain the learnability of $f : \x \mapsto \y$, we consider the continuous version. The Hilbert transform of $\mathfrak{u}$ with $\bth$ is defined by
\begin{equation}
\mathcal H_{\bth}\mathfrak{u}(\bt) =\f {1}{\pi} \int_{\R} \f{\mathfrak{u}(\bt_{\bth}(a))}{a-\bt\cdot\bth} da \end{equation}
where $\bt_{\bth}(a)$ is the point given by
\begin{equation}
\bt_{\bth}(a)=a\bth+(\bt\cdot\bth^{\perp})\bth^{\perp},~ (\bth^{\perp}=(-\sin\varphi,\cos\varphi)).
\end{equation}
Note that $\{\bt_{\bth}(a):a \in\R\}$ is the $\bth$ directional line passing through  $\bt$.
Let us define
\begin{equation}
\mR^*_{\bth_0}h(\bt)=\int_{\varphi_0}^{\varphi_0+\pi}h(\bt\cdot\bth,\varphi)d\varphi
\end{equation}
Most interior tomography algorithms are based on the following identity \cite{Noo2004,Wang2013} :
\begin{equation} \label{eq1}
\mathfrak{u}(\bt_{\bth_0}(a)) = \f{1}{2}\mR^*_{\bth_0}\mathcal H_{\bth_0}\dfrac{\p}{\p a} \mathfrak{b} (\bt_{\bth_0}(a))
\end{equation}
Applying the Hilbert transform to both sides of the above identity,
\begin{equation}
\mathcal H_{\bth_0}\mathfrak{u}(\bt_{\bth_0}(a)) =-\f{1}{2\pi} \mR^*_{\bth_0}\dfrac{\p}{\p a} \mathfrak{b} (\bt_{\bth_0}(a))
\end{equation}

Given $\bt\in \Omega_{\mbox{\tiny ROI}}$ and $\bth_0$, we have the following identity: For $\bt_{\bth_0}(a)\in \Omega_{\mbox{\tiny ROI}}$,
\begin{equation}
\mathfrak{u}(\bt_{\bth_0}(a))=\Psi^{\mbox{\footnotesize in}}_{\bth_0}\mathfrak{b}(\bt_{\bth_0}(a))+\Psi^{\mbox{\footnotesize out}}_{\bth_0}\mathfrak{b} (\bt_{\bth_0}(a))
\end{equation}
where
\begin{equation}
\Psi^{\mbox{\footnotesize in}}_{\bth_0}\mathfrak{b}(\bt_{\bth_0}(a))=\dfrac{1}{2\pi^2} \int_{\bt_{\bth_0}(a')\in \Omega_{\mbox{\tiny ROI}}} \dfrac{\mR^*_{\bth_0} \dfrac{\p}{\p a'}\mathfrak{b}(\bt_{\bth_0}(a'))}{a'-a}da'
\end{equation}
and \begin{equation}
\Psi^{\mbox{\footnotesize out}}_{\bth_0}\mathfrak{b}(\bt_{\bth_0}(a))=\dfrac{1}{2\pi^2} \int_{\bt_{\bth_0}(a')\notin \Omega_{\mbox{\tiny ROI}}} \dfrac{\mR^*_{\bth_0} \dfrac{\p}{\p a'}\mathfrak{b}(\bt_{\bth_0}(a'))}{a'-a}da'.
\end{equation}
The main point is that $\Psi^{\mbox{\footnotesize out}}_{\bth_0}\mathfrak{b}(\bt_{\bth_0}(a))$ is analytic in the line segment $\gamma=\{a\in\R:\bt_{\bth_0}(a)\in \Omega_{\mbox{\tiny ROI}}\}$ \cite{Wang2013}. This means that $\Psi^{\mbox{\footnotesize out}}_{\bth_0}\mathfrak{b}(\bt_{\bth_0}(\gamma))$ is completely determined by its knowledge in any open subset of $\gamma$. Consequently, $\mathfrak{u}(\bt_{\bth_0}(\gamma))$ can be recovered from $\Psi^{\mbox{\footnotesize in}}_{\bth_0}\mathfrak{b}(\bt_{\bth_0}(\gamma))$ and the information of $\Psi^{\mbox{\footnotesize out}}_{\bth_0}\mathfrak{b}(\bt_{\bth_0}(\gamma))$ in the small open subset of $\gamma$.

Now, we revisit the original discrete problem \eqref{LCF}. Finding the function $f:\x\mapsto\y$ is equivalent to finding the correction of the residual $\x-\y$. The analytic property of $\Psi^{\mbox{\footnotesize out}}_{\bth_0}\mathfrak{b}(\bt_{\bth_0}(\gamma))$ explains the structure of the residual $\x-\y$ in $\Omega_{\mbox{\tiny ROI}}$. Owing to the analytic structure of $\x-\y$ along the line segments,  $\x-\y$ in $\roi$  has very different image structure from that of $\y$ (medical image); therefore, $\x$ can be decomposed into $\y$ and $\x-\y$. Hence, the function $f:\x\mapsto\y$ is learnable.

\subsection{Sparse-view CT}\label{sec-sparseCT}
\begin{figure*}[t!]
\begin{center}
	\includegraphics[width=0.8\textwidth]{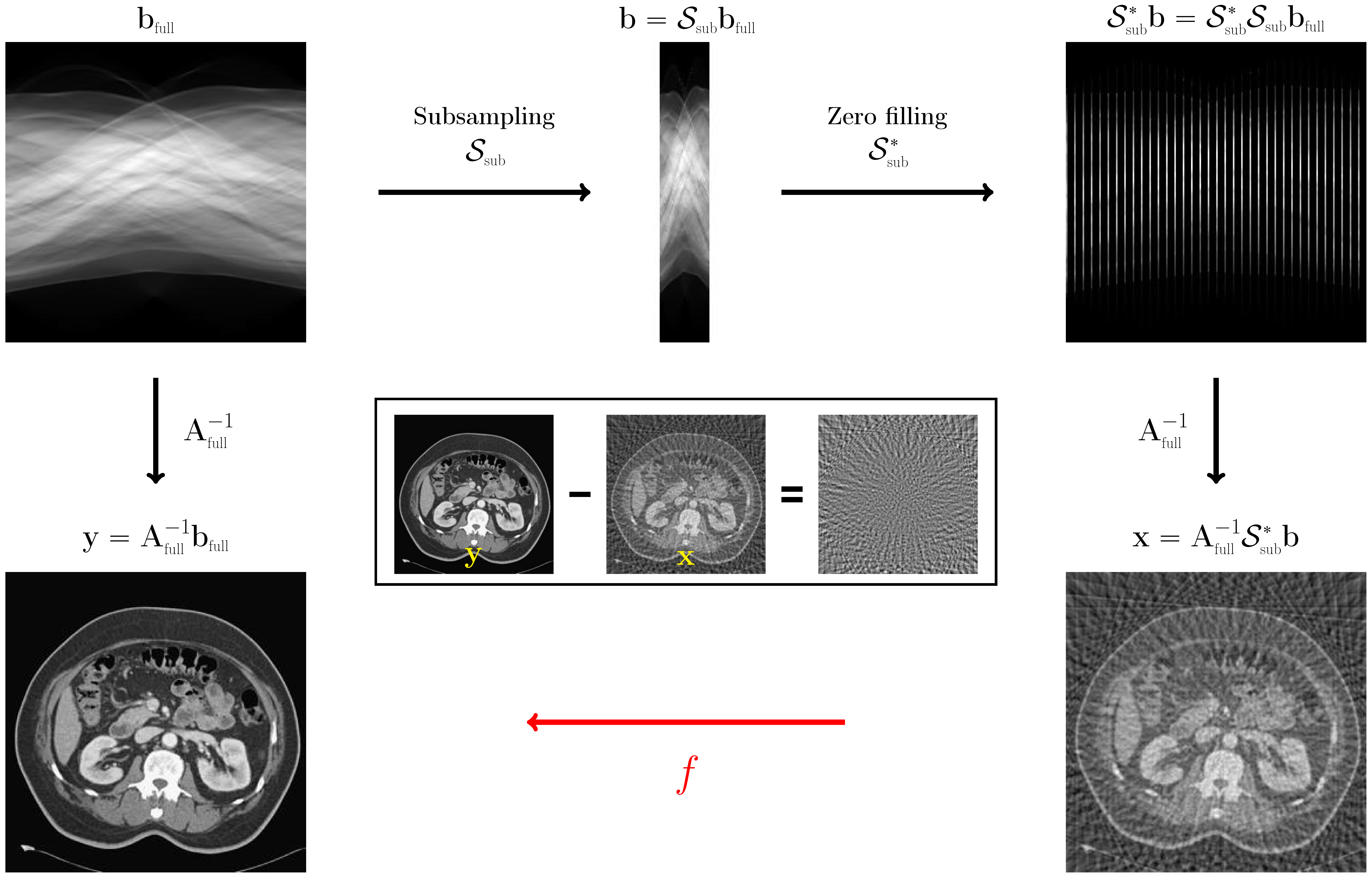}
	\caption{Sparse-view CT aims to reconstruct an image $\y=\A_{\mbox{\tiny full}}^{-1} \b_{\mbox{\tiny full}}$ from uniformly undersampled data $\b=\mathcal{S}_{\mbox{\tiny sub}}\b$, where $\mathcal{S}_{\mbox{\tiny sub}}$ is the subsampling operator. Applying the deep learning method, we attempt to learn $f$ that produces $\y$ from the input $\x=\A^{-1}_{\mbox{\tiny full}} \mathcal{S}^{*}_{\mbox{\tiny sub}} \b$. The map $f$ is learnable due to a simple structure of the residual $\x - \y$.}
	\label{SVCT}
\end{center}
\end{figure*}
The sparse-view CT problem aims to find a reconstruction function that maps from a sparse-view sinogram $\b$ to an image whose quality is as high as that of a regular CT image reconstructed by full-view sinogram $\b_{\mbox{\tiny{full}}}$. Throughout this section, we will denote the sub-sampling operator by $\mathcal S_{\mbox{\tiny sub}}$, so $\b=\mathcal \mathcal S_{\mbox{\tiny sub}}\b_{\mbox{\tiny{full}}}$.

Assuming that the subsampled data $\b=\mathcal \mathcal S_{\mbox{\tiny sub}} (\b_{\mbox{\tiny{full}}})$ violates the Nyquist's rule, the standard FBP algorithm using $\b$ produces a streaking artifacted CT image, which can be expressed as
\begin{equation}
\x=\A_{\mbox{\tiny full}}^{-1} \mathcal{S}_{\mbox{\tiny sub}}^*\b
\end{equation}
where $\mathcal{S}_{\mbox{\tiny sub}}^*\b$ is a zero-filled data of $\b$ and $\mathcal{S}_{\mbox{\tiny sub}}^*$ is the dual operator of $\mathcal S_{\mbox{\tiny sub}}$.

The corresponding high quality image reconstructed from $\b_{\mbox{\tiny{full}}}$ (satisfying the Nyquist's rule) is given by
\begin{equation}
\y=\A_{\mbox{\tiny full}}^{-1} \b_{\mbox{\tiny{full}}}
\end{equation}
The goal is to learn the function $f : \x \mapsto \y$ using $\{(\x^{(k)},\y^{(k)})\}_{k=1}^{\mathfrak{n}_{\mbox{\tiny data}}}$.

The image structures of $\x-\y$ and $\y$ are very different from each other, as shown in Figure \ref{SVCT}. Numerous researches have been conducted on image enhancement methods by suppressing noise $\x-\y$. The following CS technique is widely used to alleviate the noise $\x-\y$  :
\begin{equation}\label{TV-min2}
\y=\underset{\textbf{y}}{\mbox{argmin}} ~ \|\A^\sharp\A \y - \x \|_{\ell^2}^2 + \lambda \| \nabla \y \|_{\ell^1}
\end{equation}
where $\| \nabla \y \|_{\ell^1}$ is used to penalize the undesired feature $\x-\y$. According to \cite{Candes2005,Candes2006b}, the convex minimization problem \eqref{TV-min2} is somehow close to the following problem:
\begin{equation}\label{L0-min}
\y=\underset{\textbf{y}\in W_\alpha }{\mbox{argmin}} ~ \|\A^\sharp \A \y - \x \|_{\ell^2}^2
\end{equation}
where $W_\alpha:=\{\y: \| \nabla \y\|_0 \le \alpha\}$, $\alpha$ is a positive integer, and $\|\nabla\y\|_0$ indicates the number of non-zero entries of $\nabla\y$. Since $W_\alpha$ is a finite union of the $\alpha$-dimensional space, the constraint $\textbf{y}\in W_\alpha$ shrinks the domain of solutions by enforcing sparsity. If $\alpha$ is sufficiently smaller than $m$ (the number of equations), then $\A$ satisfies $\alpha$-RIP condition \cite{Candes2006b} within the sparse set $W_\alpha$. Namely, $\|\textbf{A}\textbf{y}-\A\textbf{y}'\|\neq 0$ for any different images $\y \neq \y'$ in $W_\alpha$, so that the uniqueness of the problem \eqref{L0-min} can be guaranteed within the sparse set $W_\alpha$.

Assuming that $\A$ satisfies the $\alpha$-RIP condition and that $\mathcal{M}_{\mbox{\tiny image}}$ is given by $W_\alpha$,  $\A$ satisfies the $\mathcal{M}$-RIP condition in \eqref{MRIP}, and, thus, the reconstruction $f$ is learnable from Observation \ref{obsMRIP}. In the sparse-view CT problem, as several CS methods exhibit fairly successful reconstruction results \cite{Zhu2013,Kudo2013}, the forward operator $\A$ seems to possess some property somewhat closely related with RIP condition on the sparse set $W_\alpha$. However, the handmade set $W_\alpha$ as prior knowledge can be viewed as a very rough approximation to the manifold $\mathcal{M}_{\mbox{\tiny image}}$, and hence it is limited in its ability to preserve small details with important medical information. Owing to the highly curved structure of $\mathcal{M}_{\mbox{\tiny image}}$ as observed in Section \ref{M-example}, deep learning approaches \cite{Jin2017,Han2018,Xie2018} have been proposed to facilitate the machine-learned intrinsic regularizer using training data.

\begin{remark}
Let us give a brief comment on $\A_{\mbox{\tiny full}}^{-1}$. In general, the solution of $\A_{\mbox{\tiny full}}\y =\b$ can be expressed as $\y = (\A_{\mbox{\tiny full}}^* \A_{\mbox{\tiny full}})^{-1} \A_{\mbox{\tiny full}}\b$ and therefore $\A_{\mbox{\tiny full}}^{-1}$ should be understood as $(\A_{\mbox{\tiny full}}^* \A_{\mbox{\tiny full}})^{-1} \A_{\mbox{\tiny full}}$. In CT, $\A_{\mbox{\tiny full}}^{-1}$ represents the operator  corresponding to the FBP algorithm as a practically feasible reconstruction.
\end{remark}

\begin{figure*}[t!]
\begin{center}
	\includegraphics[width=1\textwidth]{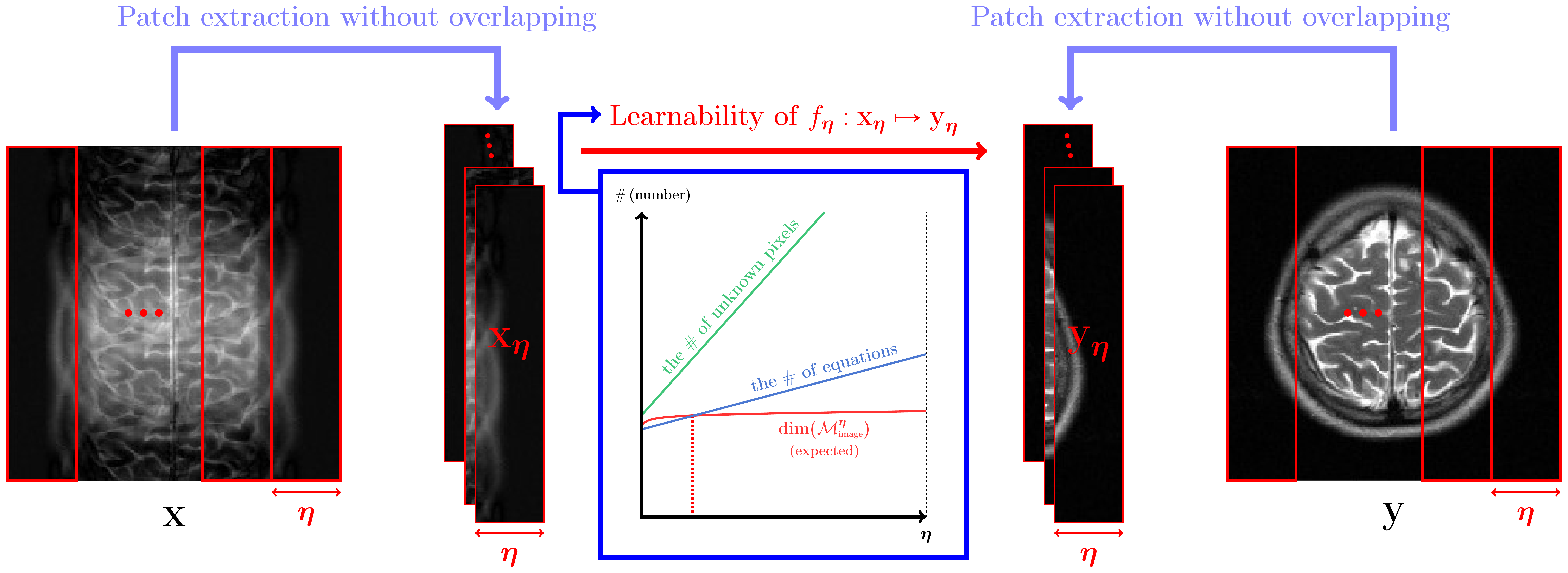}
	\caption{Learnability of $f_{\eta} : \x_{\eta} \mapsto \y_{\eta}$.  Let $\y_{\eta}$ denote an image patch of size $256 \times \eta$ extracted from a $256 \times 256$ MR image $\y$ and $\x_{\eta}$ be an aliased image obtained by $\x_{\eta} = \A^\sharp \A \y_{\eta}$. The learnability of reconstruction map $f_{\eta}$ seems to be related with the dimension of solution set $\mathcal{M}_{\mbox{\tiny image}}^{\eta}$. Here, $\mathcal{M}_{\mbox{\tiny image}}^{\eta}$ is the set of all extracted patches from $256 \times 256$ MR images.}
	\label{addscheme}
\end{center}
\end{figure*}
\section{Discussion}\label{sec-discussion}
In this section, we discuss several interesting issues related to deep learning-based solvability for underdetermined problems in medical imaging. An important question is ``what is the minimum ratio of undersampling to provide guarantee of accurate reconstruction?". It is closely related to the dimension of the manifold.

\subsection{Solvability Issue}
This subsection discuss an interesting characteristic of the learning problem in the underdetermined MRI described in Section \ref{SolveUMRI}, where the uniform subsampling of factor 4 with additional phase encoding lines is used as the subsampling strategy. For the ease of explanation, we assume $n = 256 \times 256$ (i.e. $\y$ represents a $256 \times 256$ MR image) and uniform subsampling of factor 4 adding 12 supplementary phase encoding lines.

\begin{figure*}[t!]
\begin{center}
	\includegraphics[width=0.9\textwidth]{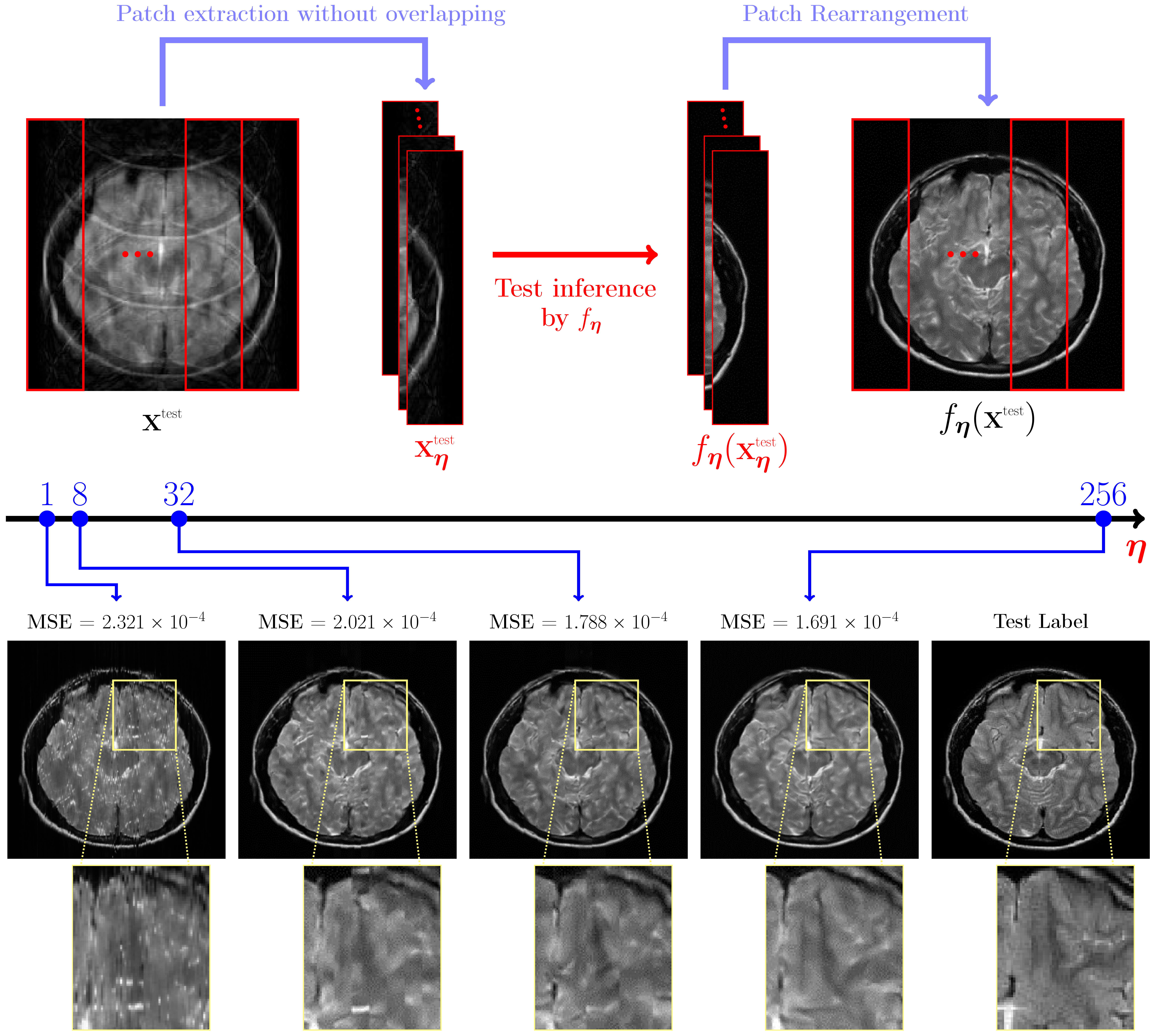}
	\caption{Test performance evaluation of $f_{\eta}$ with various $\eta = 1, 8, 32,$ and $256$. For each $\eta$, the function $f_{\eta}$ is obtained by training the  U-net with the corresponding image patches extracted from 1500 MR images $\{ \y^{(k)} \}_{k=1}^{1500}$. To compare quantitative performances with one another, a given test image $\x^{\mbox{\tiny test}} \notin \{ \x^{(k)} \}_{k=1}^{1500}$ is divided into image patches, where each patch is reconstructed through the trained U-net $f_{\eta}$, and the reconstruction outputs are rearranged into one image $f_{\eta}(\x^{\mbox{\tiny test}})$. We also qualitatively evaluate the test result qualitatively by computing the mean squared error(MSE) between the inference output $f_{\eta}(\x^{\mbox{\tiny test}})$ and label $\y^{\mbox{\tiny test}}$.}
	\label{addevaluationqa}
\end{center}
\end{figure*}

For a given integer $\eta \geq 1$, let $\{ \textbf{y}_{\eta}^{(j)} \in \mathbb{R}^{256 \times \eta} \}_{j=1}^{\mathfrak{n}_{\mbox{\tiny patch}}}$ be a set of the image patches extracted from an image $\y$ and let $\{ \textbf{x}_{\eta}^{(j)} \}_{j=1}^{\mathfrak{n}_{\mbox{\tiny patch}}}$ be the corresponding set of the aliased images, given by $\x_{\eta}^{(j)} = \A^\sharp \A \y_{\eta}^{(j)}$. We assume that there is no overlap between all the patches. (See Fig. \ref{addscheme}.)

This section aims to investigate whether the factor of $\eta$ is important in learning $f_{\eta} : \x_{\eta} \mapsto \y_{\eta}$.
To observe the effects of $\eta$, we train the U-net by varying $\eta$, using the following training dataset:
\begin{equation}
\{ (\textbf{x}_{\eta}^{(j,k)}, \textbf{y}_{\eta}^{(j,k)}) ~ | ~ j = 1,\cdots, \mathfrak{n}_{\mbox{\tiny patch}} \mbox{ and } k = 1, \cdots, \mathfrak{n}_{\mbox{\tiny data}}\}
\end{equation}
where $\textbf{y}_{\eta}^{(j,k)}$ represents $j$-th image patch extracted from the $n$-th label MR image $\textbf{y}^{(k)}$.

The reconstruction map $f_{\eta}$ aims to solve the linear system that has $76 \times \eta$ number of equations with $256 \times \eta$ number of unknowns. As $\eta$ increases, the number of unknowns increases more rapidly than the number of equations. See the middle box in Fig. \ref{addscheme}. However, our experimental results, described in Fig. \ref{addevaluationqa}, demonstrate that the learning ability is gradually improved as $\eta$ increases.

The experimental results in Fig. \ref{addevaluationqa} can be explained by means of the dimensionality of the manifold given by
\begin{equation}
\mathcal{M}_{\mbox{\tiny image}}^{\eta} := \{ \textbf{y}_{\eta} ~ | ~ \textbf{y}_{\eta} \mbox{ is a }  256 \times \eta \mbox{ patch extracted from } \textbf{y} \}
\end{equation}
The dimension of $\mathcal{M}_{\mbox{\tiny image}}^{\eta}$, denoted by $g_{\mathcal M}(\eta)$, can be viewed as a function of $\eta$ variable. As shown in the middle box in Fig. \ref{addscheme}, $g_{\mathcal M}(\eta)$ seems to grow very slowly; therefore, $g_{\mathcal M}(\eta)$ might intersect with the linear function $g_{\mbox{\tiny $\#$equations}}(\eta)=76\eta$ (i.e. the number of equations). Assuming $\eta^*$ is the intersection point (i.e. $g_{\mathcal M}(\eta^*) = g_{\mbox{\tiny $\#$equations}}(\eta^*)$), $f_{\eta}$ can be regarded as learnable, provided $\eta \geq \eta^*$. This interpretation can be supported by the error estimations in Fig. \ref{addevaluationqa}.

\begin{figure*}[t!]
\centering
\includegraphics[width=0.9\textwidth]{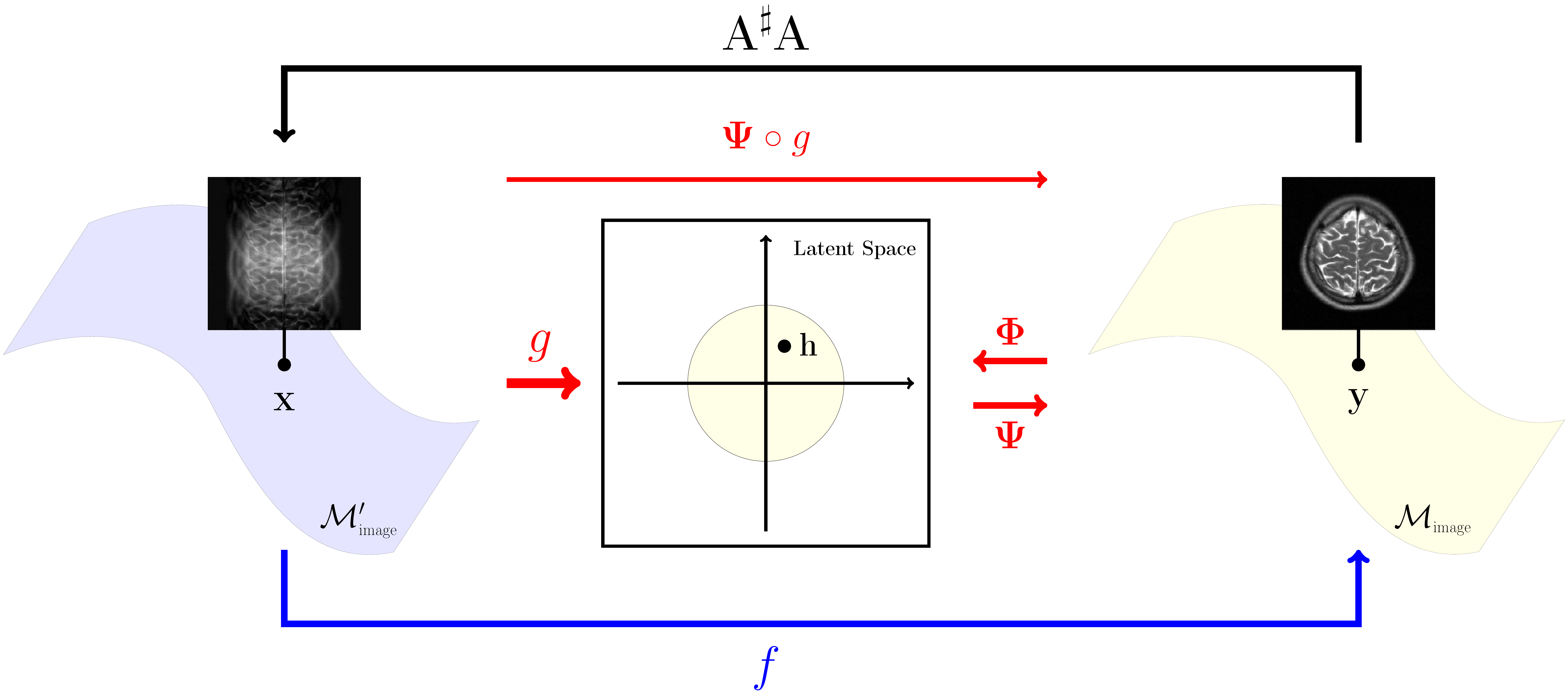}
\caption{What happens if a low dimensional latent representation is possible? Imagine that we have a low dimensional latent generator $\Psi : \h \mapsto \y$ and an encoder $\Psi : \y \mapsto \h$ such that $\Psi \circ \Phi(\y) \approx \y$ for all $\y \in \mathcal{M}_{\mbox{\tiny image}}$. If we have a map $g : \x \mapsto \h = \Psi(\y)$, then the reconstruction map $f : \x \mapsto \y$ is given by $f = \Psi \circ g$.}
\label{ManiLearn1}
\end{figure*}

Let us explain the reasons for expecting $g_{\mathcal M}(\eta)$ to grow significantly slowly as  $\eta$ increases. Assume that $\mathcal{M}_{\mbox{\tiny image}}$ is the set of all the human head MR images. Then, all the images in $\mathcal{M}_{\mbox{\tiny image}}$ possess a similar anatomical structure that consists of skull, gray matter, white matter, cerebellum, among others. In addition, every skull and tissue in the image have distinct features that can be represented nonlinearly by a relatively small number of latent variables, and so does for the entire image. Notably, the skull and tissues of the image are spatially interconnected, and even if a part of the image is missing, the missing part can be recovered with the help of the surrounding image information. This is the reason that image inpainting techniques \cite{Bertalmio2000} have been successful in image restoration for filling-in the missing areas in images. These observations seem to indicate that $g_{\mathcal M}(\eta)$ does not change much with $\eta$ near $\eta=256$, where $g_{\mathcal M}(256)$ corresponds to the dimension of the entire image. Therefore, we expect that $g_{\mathcal M}(256) \ll g_{\mbox{\tiny $\#$equations}}(256)$, so that there exists $\eta^*$ (the turning point for learnability) such that $g_{\mathcal M}(\eta^*) = g_{\mbox{\tiny $\#$equations}}(\eta^*)$. If the curse-of-dimensionality does not matter, it may be better to learn 3D images in total rather than dividing 3D images into multiple pieces. A rigorous mathematical analysis of this issue  is the subject of our future research topic.

\subsection{Some issues on learning low dimensional representation} \label{LDR}
In our undersampled problems, the dimension of $\y$ (i.e., the total number of pixels in the image) is considerably bigger than the dimension of $\b$ (i.e., the number of independent components in the measurement data). Since there exist infinitely many images that solve the mathematical model $\A\y=\b$, we need to reflect prior information about the unknown solution manifold, either implicitly or explicitly, in the image restoration process.

Over several decades, various regularization approaches have been used with predefined convex regularization functionals in order to incorporate a-priori information on $\y$. CS methods involving $\ell^1$-norm regularization minimization have been powerful for noise removal, whereas they suffer from limitations in preserving small features. In medical imaging, there are a variety of small features, such that the difference in the data fidelity is very small as compared with that in normalization, whether or not those small features are present. Hence, finding a more sophisticated normalization to keep small features remains a challenging problem.

Unlike in CS (predetermined convex-norm-based approach), deep learning can be viewed as a black box model approach where training data is used for probing the solution manifold. To ensure the possibility of solving undersampled problems through deep learning, it would be desirable to investigate the performance of low dimensional representation learning for the unknown manifold from training data.

Autoencoder(AE) techniques (as the natural evolution of PCA) are widely used to find a low dimensional representation for the unknown $\mathcal M_{\mbox{\tiny image}}$ from the known training data $\{\y^{(k)}\}_{k=1}^{\mathfrak{n}_{\mbox{\tiny data}}}$ \cite{Hinton2006,Kingma2013}. The AE consists of a encoder $\Phi:\y\to \h$ for a compressed latent representation
and a decoder $\Psi:\h\to \y$ for providing $\Psi\circ \Phi(\y)\approx \y$ (i.e., an output image is similar to the original input image). Assuming that $\Phi$ provides a satisfactory approximation of the solution manifold, the underdetermined problem \eqref{undersampled} can be solved as follows:
\begin{equation}\label{min-Psi}
f(\x)=\Psi(\h),~~ \h=\underset{\h}{\mbox{argmin}} ~ \| \A\Psi(\h) -\A\x\|
\end{equation}
A deep learning technique can be used to solve the problem \eqref{min-Psi}, as in the minimization problem \eqref{min-Psi} it might be difficult to use the standard gradient descent method because of the complex deep learning structure of the decoder $\Psi$. If the dimension of the latent space is reasonably small, the reconstruction map is achieved by
\begin{equation}\label{min-Psi2}
f(\x)= \Psi\circ g (\x),~g:=\underset{g\in\Bbb{NN}}{\mbox{argmin}} \sum_{k=1}^{\mathfrak{n}_{\mbox{\tiny data}}} \| g({\x}^{(k)}) -\h^{(k)}\|
\end{equation}
where $\h^{(k)}= \Phi(\y^{(k)})$. One can refer to Fig. \ref{ManiLearn1} for the schematic understanding of the method. Recent papers have reported that the AE-based approaches show remarkable performances in several applications \cite{Chang2017,Seo2019,Jalali2019,Tezcan2018}. However, for high dimensional data, AEs seem to suffer from the blurring and loss of small details, as depicted in Fig. \ref{ManiLearnResult1}. Improving performance of AEs in high dimensional medical image applications is still a challenging issue.

\begin{figure*}[t!]
	\centering
	\includegraphics[width=1\textwidth]{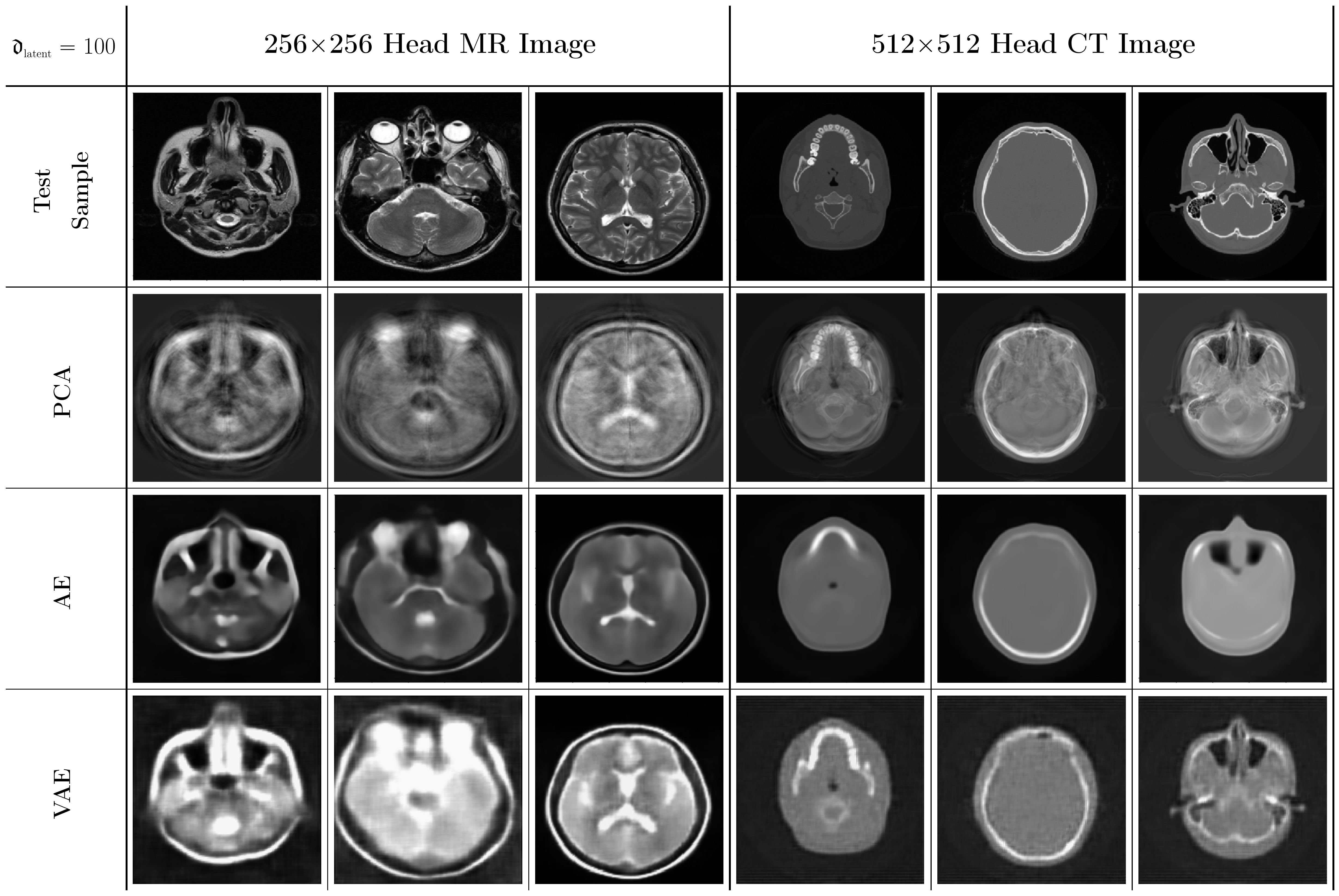}
	\caption{Learning a low dimensional representation for MR and CT images with three different dimension reduction techniques (principal component analysis(PCA), auto-encoder(AE), and variational AE(VAE)).
		By using each technique, an encoder $\Phi$ and decoder $\Psi$ are trained so that MR or CT images are projected into a 100 dimensional space (i.e. $\mathfrak{d}_{\mbox{\tiny latent}}=100$) by using 1800 MR images of pixel dimension $256 \times 256$ or using 3400 CT images of pixel dimension $512 \times 512$.
		PCA was performed by a built-in function \texttt{pca} in MATLAB. For AE and VAE, we used convolutional AE and VAE structures, respectively, consisting of modified residual blocks from ResNET \cite{He2016a,He2016b}.
		The networks were implemented in the Tensorflow environment. After the training, we tested three different samples for each MR and CT case, displayed in the first row, and the corresponding test results are displayed from the second to last row.}
	\label{ManiLearnResult1}
\end{figure*}
sGenerative adversarial networks (GANs) have been utilized to generate realistic images via interactions between learning and synthesis \cite{Goodfellow2014a,Radford2016,Yi2019}. Typically, the architecture of GANs comprises two main parts; generator $\Psi$ and discriminator $\Gamma$. The GAN network aims to find a $\Psi$ that maps from a random noise vector $\h$ in the latent space to an image in a real data distribution associated with the training data $\{\y^{(k)}\}_{k=1}^{\mathfrak{n}_{\mbox{\tiny data}}}$. The generator $\Psi$ is trained with the assistance of the discriminator $\Gamma$ in such a way that $\Gamma$ misclassifies $\Psi(\h)$ as a real image. The training procedure of GAN can be viewed as a performance competition between the generator and the discriminator.
\begin{figure*}[t!]
\centering
\includegraphics[width=1\textwidth]{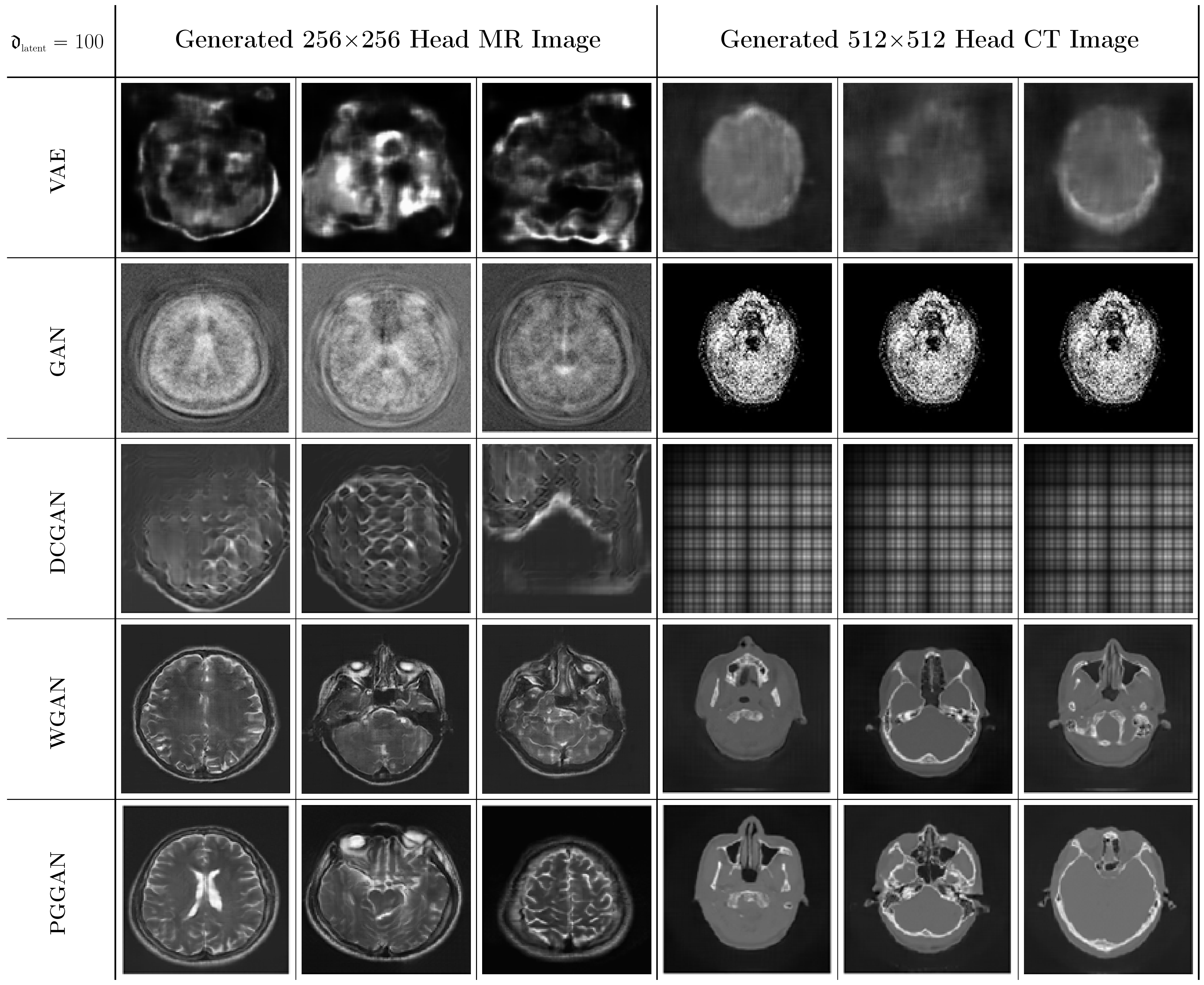}
\caption{Learning low dimensional representation for MR and CT images using five different generative models, variational auto-encoder(VAE), generative adversarial network(GAN), deep convolutional GAN(DCGAN), Wasserstein GAN(WGAN), and progressive generative GAN(PGGAN). Using these generative models, we learn a generator $\Psi$ that synthesizes MR or CT images from points sampled in 100 dimensional Gaussian distribution (i.e. $\mathfrak{d}_{\mbox{\tiny latent}}=100$). In our experiments, we used a typical network structure for each generative model, as described in \cite{Kingma2013,Goodfellow2014a,Radford2016,Gulrajani2017,Karras2018}. We replaced the cross-entropy of GAN into Hinge loss, and the Wasserstein loss of PGGAN into Hinge loss with spectral normalization in order to obtain the stability on learning \cite{Zhang2018}. All the networks were implemented in the Tensorflow environment with 1800 MR images and 3400 CT images. In each row, three different synthesized images from each trained generative models are displayed for each CT and MR case.}
\label{ManiLearnResult2}	
\end{figure*}
Although GANs have achieved remarkable success in generating various realistic images, there exist some limitations in synthesizing high resolution medical data. The GAN's approach makes it difficult to deal with high-dimensional data because the generated image can be easily distinguished from the training data, which can lead to collapse or instability during the training process \cite{Odena2017}. Several variations of GAN, such as Wasserstein GAN(WGAN) \cite{Arjovsky2017,Gulrajani2017} and progressive growing GAN(PGGAN) \cite{Karras2018}, have been developed to deal with the training instability. WGAN uses Wasserstein distance, which may improve the loss sensitivity with respect to change of parameters, compared to the Jensen-Shannon distance used in the original GAN. Fig. \ref{ManiLearnResult2} shows that in our WGAN experiment, nearly plausible synthesis results are generated for high dimensional medical image, whereas the synthesized images still suffer from somewhat lack of reality. PGGAN facilitates synthesis of high dimensional data via hierarchical multi-scale learning fashion from low resolution to the desired high resolution; therefore, the network can focus on the overall structures at the beginning of the process, before shifting attention gradually to finer scale details via later connections as the training advances. Fig. \ref{ManiLearnResult2} depicts our PGGAN experiment with the spectral weight normalization \cite{Miyato2018}, that provides high quality synthesis.

Unfortunately, unlike bidirectional AE, which provides an explicit prior, GANs learn only the unidirectional mapping $\Psi : \h \mapsto \y$ and, therefore, it is difficult to achieve \eqref{min-Psi2}. Recently, new strategies have been developed to use the implicit prior obtained using GANs as a solution prior for solving ill-posed inverse problems such as the undersampled MRI \cite{Narnhofer2019}, image denoising \cite{Tripathi2018}, and inpainting \cite{Yeh2017}. Also, a scattering generator \cite{Angles2018} with the advantages of both GAN and AE has been proposed. However, applications in high dimensional medical imaging problems are still far from satisfactory.

\section{Conclusion}
This work concerns with the solvability of undersampling problems with the use of training data. The undersampled MRI, sparse view  CT, and interior tomography are typical examples of the underdetermined problem, which involve much fewer equations (measured data) than unknowns (pixels of the image). To compensate for the uncertainty of the huge number of free parameters (difference between the number of equations and the number of unknowns), we need to limit the solution manifold using prior information of the expected images. Regularization techniques have been widely used to impose very specific prior distributions on the expected images, such as penalizing a special norm of images (or promoting sparsity in expressions). However, norm-based regularization might not actually be able to provide a clinically useful image properly in advance. Well-known CS methods, which employ random sampling and are based on regularization methods, are effective in alleviating highly oscillatory noise while maintaining the overall structure; however, sparse sensing techniques tend to eliminate small anomalies, as shown in Section \ref{CSsection}.

DL techniques appear to deal with various underdetermined inverse problems by effectively probing the unknown nonlinear manifold $\mathcal{M}_{\mbox{\tiny image}}$ through training data. It seems to handle the uncertainty of solutions to the highly ill-posed problems.

According to Hadamard \cite{Hadamard1902}, the linear problem $\A\y =\b$ is well-posed if the following two conditions hold (while ignoring the existence issue): first, for each $\b$, it has a unique solution, and second, the solution is stable under the perturbation of $\b$.  However, we note that whether the problem is well-posed depends on the choice of the solution space. Many ill-posed problems can be well-posed within the constrained solution spaces (e.g. sparse solution spaces). For a simple example, we consider  the Poisson's equation $ \nabla \cdot \nabla u = \b $ in the 2-D domain $\Omega = \{ (r \cos \theta,r  \sin \theta ) ~ | ~  0 < r < 1, ~ 0 < \theta < \frac{3\pi}{2}\}$ with the homogeneous Dirichlet boundary condition $u|_{\p\Om}=0$. If $u_*$ is a solution, so are $u=u_*+ (r^{\frac{2}{3}n}-r^{\frac{-2}{3}n})\sin(\frac{2}{3}\theta)$
for $n=0,1,2,\cdots.$. Hence, the problem is ill-posed without the constraint of the Sobolev sapce $H^1(\Omega) = \{ u ~ | ~ \int_{\Omega} |u|^2 + |\nabla u|^2 < \infty \}$.  Similarly, the underdetermined problem $\A\y=\b$ can be well-posed under a suitable solution manifold  $\mathcal{M}_{\mbox{\tiny image}}$. DL methods seem to possess ambiguous capability of learning data representation.

Recently, several experiments regarding adversarial classifications \cite{Finlayson2019,Ching2018} (e.g., false positive output of cancer) have shown that deep neural networks obtained via gradient descent-based error minimization procedure are vulnerable to various noisy-like perturbations, resulting in incorrect output (that can be critical in medical environments). These experiments show that a well-trained function $f$ in \eqref{MDL} works only in the immediate vicinity of a manifold, whereas producing incorrect results if the input deviates even slightly from the training data manifold. In practice, the measured data is exposed to various noise sources such as machine dependent noise; therefore, the developed algorithm must be stable against the perturbations due to noise sources. Hence, normalization of the input data is essential for improving robustness and generalizability of the deep learning network against adversarial attacks \cite{Goodfellow2015,Yuan2019,Finlayson2019}.

For input data normalization, we attempt to project the input $\x$ to its normalized form $\mathfrak N(\x)$ in the way that two images $\x$ and $\mathfrak N(\x)$ are almost the same from the viewpoint of radiologists.
It is quite complicated to define the distance $\mbox{dist}_{\mbox{\tiny radiologist}}(\x, \x')$ in terms of radiologist's view. If we have a good generator $\mathfrak{G}$, it can be defined as
\begin{equation}\label{dist-radio}
\mbox{dist}_{\mbox{\tiny radiologist}}(\x, \x')=\| \h - \h^\prime \|
\end{equation}
where $\mathfrak G (\h) = \mathfrak N(\x)$ and $\mathfrak G(\h^\prime) = \mathfrak N (\x^\prime)$. The Euclidian distance $\| \h - \h^\prime \|$  can be somewhat equivalent to the geodesic distance between $\mathfrak N(\x)$ and $\mathfrak N(\x')$ on the manifold. For example, given a noisy input $\x$, its normalization $\mathfrak N(\x)$  can be a denoised image while preserving the salient features of $\x$. The issues of finding the normalization $\mathfrak N$ and the generator $\mathfrak G$ would be very challenging tasks. Although there are still several challenging issues in deep learning associated with solving ill-posed inverse problems in medical imaging area, recent remarkable developments indicate that it has enormous potential to provide a useful means of overcoming the limitations of the traditional methods.

\vspace{0.25cm}
\appendix

\section*{Supplementary Material}
This appendix provides deep learning results for undersampled MRI, interior tomography, and sparse-view CT problem. To learn a reconstruction map $f : \x \mapsto \y$, we adopt the modified U-net architecture, described in Fig. \ref{FigMPDL}, and set feature depths of the network to be multiples of 64.

To train U-net, we generate a training dataset $\{\x^{(k)},\y^{(k)}\}_{k=1}^{\mathfrak{n}_{\mbox{\tiny data}}}$ in the following sense: Let $\{ \y^{(k)}\}_{k=1}^{\mathfrak{n}_{\mbox{\tiny data}}}$ be a given set of medical images. Here, $\y^{(k)}$ represents $256 \times 256$ head MR image in undersampled MRI problem and $256 \times 256$ abdominal CT image in interior tomography and sparse-view CT problem. The corresponding input data $\x^{(k)}$ is generated by computing $\x^{(k)} = \A_{\mbox{\tiny full}}^{-1} \mathcal S^*_{\mbox{\tiny sub}} \A \y^{(k)}$. In our real implementation, 1500 MR images \cite{Loizou2010} are used for undersampled MRI problem and 1800 CT images are used for interior tomography and sparse-view CT problem.

With the generated training dataset, U-net is trained by minimizing $\ell^2$ loss, as mentioned in \eqref{DLEQ}. the minimization process was performed using Adam optimizer with learning rate 0.001, mini-batch size 16, and 3000 epochs. In addition, Batch normalization \cite{Kingma2014} is used for mitigating the overfitting issue.
\begin{figure}[h]
	\centering
	\includegraphics[width=1\textwidth]{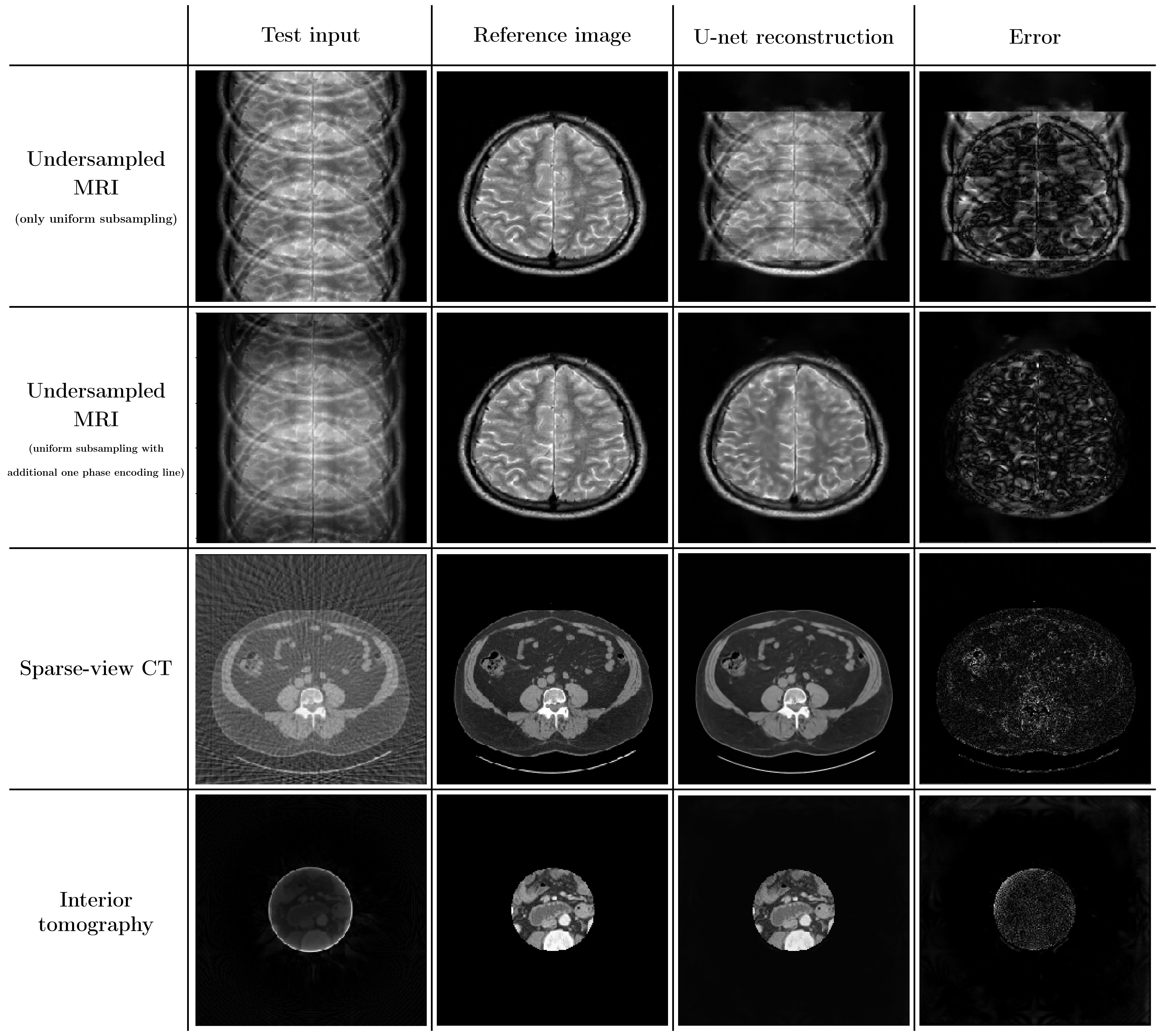}
	\caption{Deep learning results using U-net for undersampled MRI, sparse-view CT, and interior tomography problem. The images in first, second, third, and fourth column represent test input images, reference images, U-net reconstruction results, and errors, respectively.}
	\label{Appendix_Fig2}	
\end{figure}
\end{document}